\documentclass[12pt,a4paper]{amsart}
\usepackage[english]{babel}
\usepackage[ansinew]{inputenc}
\usepackage{amsmath}
\usepackage{amssymb}
\usepackage{amsthm}
\usepackage{enumerate}
\usepackage{color}
\usepackage{hyperref}
\usepackage{graphicx}

\begin{document}
	
	\theoremstyle{plain}
	\newtheorem{Thm}{Theorem}
	\newtheorem{Cor}[Thm]{Corollary}
	\newtheorem{Prop}[Thm]{Proposition}
	\newtheorem{Lem}[Thm]{Lemma}
	\theoremstyle{definition}
	\newtheorem{Def}[Thm]{Definition}
	\newtheorem{Prob}[Thm]{Problem}
	\theoremstyle{remark}
	\newtheorem{Rem}[Thm]{Remark}
	
	\title[Incomplete Worst-Case Optimal Investment]{Worst-Case Optimal Investment\\ in Incomplete Markets}
	
	%\title{Worst-Case Optimal Investment in Incomplete Markets}
	\author[Desmettre]{Sascha~Desmettre}
	\address{S.~Desmettre, Institute of Financial Mathematics and Applied Number Theory,
		Johannes Kepler University (JKU) Linz, A-4040 Linz, Austria.}
	\email{sascha.desmettre@jku.at}
	
	\author[Merkel]{Sebastian~Merkel}
	\address{S. Merkel, School of Economics,
		University of Bristol, Priory Road Complex, Priory Road, Bristol BS8 1TU, United Kingdom.}
	\email{s.merkel@bristol.ac.uk}
	
	\author[Mickel]{Annalena~Mickel}
	\address{A. Mickel, Mathematical Institute and DFG Research Training Group 1953, University of Mannheim, B6, 26, D-68131 Mannheim, Germany}
	%\email{amickel@mail.uni-mannheim.de}
	
	\author[Steinicke]{Alexander~Steinicke}
	\address{A. Steinicke, Department of Mathematics and Information Technology,
		Montanuniversitaet Leoben, A-8700 Leoben, Austria.}
	\email{alexander.steinicke@unileoben.ac.at}
	
	\date{December 16, 2024}

	\begin{abstract}
		We study and solve the worst-case optimal portfolio problem as pioneered by Korn and Wilmott in \cite{KornWilmott2002} of an investor with logarithmic preferences facing the possibility of a market crash with stochastic market coefficients by enhancing the martingale approach developed by Seifried in \cite{Seifried2010}. With the help of backward stochastic differential equations (BSDEs), we are able to characterize the resulting indifference optimal strategies in a fairly general setting. We also deal with the question of existence of those indifference strategies for market models with an unbounded market price of risk. We therefore solve the corresponding BSDEs via solving their associated PDEs using a utility crash-exposure transformation. Our approach is subsequently demonstrated for Heston's stochastic volatility model, Bates' stochastic volatility model including jumps, and Kim-Omberg's model for a stochastic excess return.  \\[3ex]
	\end{abstract}

	\maketitle

	\section{Introduction}
	
An important aspect that is neglected in the pure Merton type portfolio optimization
setting is the presence of so-called crash scenarios as first introduced by Hua and
Wilmott \cite{HuaWilmott1997} in discrete time. In this setting, parameters
are subject to Knightian uncertainty in the sense of Knight \cite{Knight1921}, which consequently does
not impose any distributional assumptions. In particular, in these worst-case optimization models, a financial crash
is identified with an instantaneous jump in asset prices.

The literature strand on worst-case portfolio optimization possess by now a long history. In their seminal work  \cite{KornWilmott2002}, Korn and Wilmott have solved the worst-case scenario portfolio problem under the threat of a crash for logarithmic utility in continuous time. Their results have been extended in \cite{KornMenkens2005}  by using the so-called indifference principle. Korn and Steffensen \cite{KornSteffensen2007} then derive (classical) HJB systems for the worst-case portfolio problem. What all these works have in common from a conceptual point of view, is that the resulting worst-case optimal strategies are characterized by the requirement that the investor is indifferent between the worst crash happening immediately and no crash happening at all.

Based on a controller-vs-stopper game, Seifried introduced fundamental concepts for the worst-case portfolio optimization in \cite{Seifried2010}, namely the indifference frontier, the indifference optimality principle and the change-of-mea\-sure device (see also \cite{KornSeifried2009}), in order to generalize the results to multi-asset frameworks and in particular discontinuous price dynamics. During the course of this paper, we will heavily rely on these methods and enhance them by allowing the market coefficients - in particular the volatility and the excess return - to be stochastic.  

Further generalizations of the worst-case approach comprise, among others, proportional transaction costs, cf.~\cite{BMS2015}, a random number of crashes, cf.~\cite{BCM2014},  lifetime-consumption, cf.~\cite{DesmettreKornSeifried2015},  a second layer of robustness, cf.~\cite{DKRS2015}, explicit solutions for the multi-asset framework, cf.~\cite{KL2019}, and more recently stress scenarios, cf.~\cite{KM2022b} and dynamic reinsurance, cf.~\cite{KM2022}.

From an abstract point of view, the worst-case approach shares common features with classical robust portfolio optimization, which typically focus on the financial markets' parameters. For instance in \cite{TalayZheng2002}, the market acts against the trader and chooses the worst possible market coefficients. Among others, another example is Schied, cf.~\cite{Schied2005}, who considers a set of probability measures to maximize the robust utility of the terminal wealth. We refer to \cite{FSW2009} for an overview on this literature. We however wish to stress that in the worst-case portfolio optimization problem, the jump times and the jump intensity are unknown, which renders the problem more delicate than standard portfolio optimization problems.

Another strand of literature, to which our work is related, is portfolio optimization with unhedegable risks, which typically comes along with incomplete markets.  The seminal paper of Zariphopoulou \cite{Zariphopoulou2001} introduces the so-called martingale distortion, which is able to deal with stochastic volatility models in a very general factor model setting.  Concerning Heston's model, Kraft \cite{Kraft2005} finds explicit solutions for power utility using stochastic control methods. Martingale methods are then employed in \cite{KMK2010} in an affine setting. Related works in that context include as well \cite{Pham2002,CV2005,Liu2007}.  In a setting with stochastic excess return, Kim and Omberg \cite{KimOmberg1996} find optimal trading strategies for HARA utility functions. For a concise overview of asset allocation  in the presence of jumps, both in the asset price and the volatility, we refer to \cite{Branger2008} and the references therein. In the context of worst-case portfolio optimization, so far market coefficients are assumed to be constant - with the exception of Engler and Korn \cite{EnglerKorn2014}, who solve the worst-case optimization problem for a Vasicek short rate process.

In this work, we combine the strands of  worst-case portfolio optimization and optimal investment with unhedgeable risks as follows:

\begin{itemize}
	\item We solve the worst-case optimal investment problem of an investor with logarithmic utility, facing both structural crashes and jump risk, in a setting that allows as well for stochastic market coefficients.
	\item We enhance the concepts indifference frontier, the indifference optimality principle and the change-of-measure device  to the case of stochastic market coefficients for logarithmic utility.  
	\item We characterize the resulting indifference strategies via the unique solutions of BSDEs, using the so-called utility crash exposure. 
	\item We exemplify and analyze the resulting indifference strategies for the Heston model, the Bates model and the Kim-Omberg model.

\end{itemize}

We also contribute to the general theory of backward stochastic differential equations, since the equation that emerges when describing indifference strategies, leads to a BSDE coefficient that does not satisfy a Lipschitz condition with deterministic constants. Rather, we are confronted with a stochastic Lipschitz constant that satisfies an exponential integrability condition. Additionally, the generator is exponentially integrable itself. In the setting without jumps, similar conditions, sufficient for existence and uniqueness, have e.g.~been treated in \cite{Briand2008}, \cite{ElKarouiHuang1997} and \cite{lijia14}. In the case including jumps, BSDEs with stochastic Lipschitz condition have been treated in \cite[Chapter II]{Saplaouras17}. However, the spaces for solutions used there are different than those in the standard theory with deterministic Lipschitz constants. Our results in this article still allow the use of the standard spaces. The approach we follow is based on an approximation argument building on the L\'evy settings used e.g.~in \cite{Kruse} and \cite{KremsnerSteinicke}. We obtain existence, uniqueness of a solution and a comparison theorem, necessary to guarantee the one-to-one relation between BSDEs and PDEs (see \cite{bbp}). This relation between the deterministic and stochastic world needs several requirements, e.g. a local Lipschitz continuity in the initial value of the stochastic process that models the volatility of the asset price. To be applicable to the popular model choice of the CIR process, we extend the existing result from \cite{CHJ} to a wider parameter range.

The remainder of the paper is organized as follows: 
In Section~\ref{sec:market} we define the financial market and and the worst-case optimization problem in incomplete markets. In Section~\ref{sec:post} and 
and Section~\ref{sec:pre}, we solve the worst-case portfolio optimization problem by disentangling the problem in the post-crash and pre-crash problem. Section~\ref{sec:indifferenceBSDE} then develops the BSDE machinery which is needed for the characterization of indifference strategies when market coefficients are stochastic. Section~\ref{sec:Markovian} and Section~\ref{sec:numerics} deal with the concrete examples, i.e. the Heston model, the Bates model and the Kim-Omberg model. Appendices~\ref{app:post}, \ref{App:Proofs} and \ref{App:Markov} contain several proofs and auxiliary results.

	\section{Financial Market Model and Worst-Case Optimization Problem}\label{sec:market}

 Let $(\Omega,\mathcal{F},\mathbb{P})$ be a probability space carrying two independent Brownian motions $\hat W, \tilde{W}$ and a Poisson random measure $\nu$ independent from $(\hat{W}, \tilde{W})$.  $(\mathcal{F}_t)_{t\in[0,T]}$ is the natural, augmented filtration for $\hat W, \tilde{W}$ and $\nu$ satisfying the usual conditions. Note that for some $\rho\in [-1,1], W:=\rho\hat{W}+\sqrt{1-\rho^2}\tilde{W}$ is again a Brownian motion.

		The financial market consists of a riskless money market account $B$ and a risky asset $S$. In the absence of a crash, the dynamics of $B$ and $S$ are given by
		\begin{align*}
		dB_t &= B_t r_t dt \,,\quad &&B_0 = 1\,,\\
		dS_t &= S_{t-} \left[ (r_t + \lambda_t) dt + \sigma_t \, dW_t -\int_{[0,l^{L}_{\max}]}l\nu(dt,dl)\right]\,,\quad &&S_0 = s\,,
		%dz_t &= \kappa(\theta - z_t)) dt + \sigma \sqrt{z_t} d\hat{W}_t \,,\quad &&z_0 = z\,,
		\end{align*}
		where $r$, $\lambda$ and $\sigma > 0$ are progressively measurable processes (w.r.t.~the filtration $(\mathcal{F}_t)_{t\in[0,T]}$) describing the dynamics of the market coefficients. Moreover, we denote the intensity measure of the Poisson random measure $\nu$ by $\vartheta$ and by $\tilde{\nu}$ we denote the according compensated Poisson random measure; we assume that $\vartheta$ is a L\'evy measure with support in $[0,l^{L}_{\max}]$ and $\int_{[0,l^{L}_{\max}]}l\vartheta(dl)<\infty$. In addition, we assume $l^{L}_{\max}<l^{WOC}$, where $l^{WOC}\in (0,1)$ is the given size of the model's worst-case substantial crash.
		 If the model does not permit jumps (i.e.~$\nu$ and $\vartheta$ do not appear), we call this setting {\it purely Brownian model} or {\it model without jumps}.

		A particularly important choice of those parameters leads to the following version of the Bates model \cite{Bates1996} and Heston model \cite{Heston1993}, respectively:
		\begin{equation} \label{eq:Bates}
		\begin{aligned}
		dB_t &= B_t r_t dt \,,\quad &&B_0 = 1\,,\\
		dS_t &= S_{t-} \left[ (r_t + \lambda_t) dt + \sqrt{z_t} \, dW_t -\int_{[0,l^{L}_{\max}]}l\nu(dt,dl)\right]\,,\quad &&S_0 = s\,,\\
		dz_t &= \kappa(\theta - z_t) dt + \tilde{\varsigma} \sqrt{z_t} d\hat{W}_t \,,\quad &&z_0 = z\,,
		\end{aligned}
		\end{equation}
		where $r$ is a progressively measurable interest rate process, $\lambda$ is a progressively measurable excess return process, $\theta > 0$ the mean reversion level of the volatility, $\kappa > 0$ the mean reversion speed of the volatility, $\tilde{\varsigma} > 0$ the volatility of volatility and $\hat{W}$ is another Brownian motion with $\left<W,\hat{W}\right>_t=\rho t$. 
		
		Alternatively, the Kim-Omberg model, cf.~\cite{KimOmberg1996}, is constituted by an analogous set of equations, where now $\lambda_t=z_t$ with constant $\sigma_t=\tilde{\sigma}\in (0,\infty)$. We elaborate on the details in Section~\ref{sec:Markovian}.

		%with constants $\theta, \kappa, \tilde{\varsigma}$. Furthermore, we have 
%		\textcolor{magenta}{
%		\begin{Rem}
%		Note that the above is equivalent to a formulation with two independent Brownian motions $\hat W$ and $\tilde W$ by setting 
%		\begin{align}
%		W := \rho \hat W + \sqrt{1-\rho^2} \tilde W\,.
%		\end{align} 	
%		\end{Rem}}	
%		
%		In addition, to simplify matters further here, it is assumed that short-sales are prohibited.
%		
	We face two types of crashes in our model. The number $l^{WOC} \in (0,1)$ describes the structural worst-case crash associated with a catastrophe. By $l$ we denote continuously occurring crashes of moderate size with $l \in (0,l^{L}_{\max}]$, where $l^{L}_{\max} < l^{WOC}$; this is line with the reasoning of \cite{Seifried2010}. 
		
    %We record the detailed definition of a worst-case crash scenario as follows: 
\begin{Def}
	For a fixed crash height $l^{WOC}$, a worst-case crash scenario $(\tau,l^{WOC})$ is given by a $[0,T]\cup\{\infty\}$-valued $(\mathcal{F}_t)_{t\in [0,\infty)}$-stopping time
	$\tau$. The stock index dynamics
	in the crash scenario $(\tau,l^{WOC})$ is given by
	\begin{align}
	%dP_t &= S_t[(r + \lambda)dt + \sigma dW_t] \mbox{ on } [0, \tau) \mbox{ and } [\tau, T]\,, \quad S_0 = S,\\
	S_\tau &= (1 - l^{WOC})S_{\tau -}.
	\end{align}
	The no-crash scenario corresponds to $\tau = \infty$. We denote by $\Theta$ the collection of all stopping times $\tau$. 
\end{Def}
		
		%The definition of a worst-case crash scenario is solely characterized by the crash time and the set of all crash scenarios, i.e. $[0,T]\cup\{\infty\}$-valued stopping times, is denoted by $\Theta$. } 
		
		\begin{Rem}
		In this paper we restrict attention to the case of a single market
		crash. This is in line with e.g.~the analysis in \cite{DesmettreKornSeifried2015}. The model can be generalized to an arbitrary finite number of crashes. In that case the optimal strategies can be determined recursively. Moreover, it is also possible to treat the crash height $l^{WOC}$ as a bounded random variable; however this is mainly notational, since the worst case is anyway attained by its upper bound. % see, for instance,
		%the related discussions in Korn & Steffensen (2007) and Korn & Wilmott (2002) or
		%Seifried (2010).
	\end{Rem}

	Throughout the paper we make the following general standing assumption on the market coefficients $r$, $\lambda$ and $\sigma$:
		\begin{equation}\tag{A}\label{assumption:stochCoeff:CoefficientsContinuous}
		\begin{aligned}
			&\sigma^2>0, \lambda,\mbox{ and } \sigma\, \mbox{ are continuous and  progressively measurable w.r.t.~} \\
			&(\mathcal{F}^{\hat{W},\tilde{W}}_t)_{t\in [0,\infty)}(\mbox{the augmented filtration generated by } \hat{W}\mbox{ and }\tilde{W}),\\
			&r \text{ is } (\mathcal{F}_t)_{t\geq 0}\text{-progressively measurable, }\mathbb{E}\left[\int_0^T|r_t|dt\right]<\infty.
		\end{aligned}
		\end{equation}
		 
		\begin{Rem}\label{Rem:JumpsExcludeRemark}
		The second line of \ref{assumption:stochCoeff:CoefficientsContinuous} means that the source of jumps is not involved in the excess return $\lambda$ and the volatility $\sigma$. However, the jumps do enter the model in the equation for the wealth process below.  
				\end{Rem}
		In addition, we need the following integrability assumptions on the market coefficients $\lambda$ and $\sigma$:
 
		\begin{align} 
			&\mathbb{E}\left[\int_0^T{\left(|\lambda_t|+|\sigma_t|^2\right)dt}\right]<\infty \,. \tag{B1}\label{assumption:stochCoeff:integrability}
	\end{align}

	For future reference we define the function
	$$\Phi_t:\Omega \times [0,\infty)\rightarrow \mathbb{R}\,,\quad (\omega,y)\mapsto r_t(\omega) + \lambda_t(\omega) y - \frac{1}{2}\sigma^2_t(\omega) y^2 +\int_{[0,l^{L}_{\max}]}\log(1-yl)\vartheta(dl).$$
	
	\paragraph{\textit{Admissible portfolio processes}} We restrict our attention to admissible portfolio processes with continuous paths.
	
	\begin{Def}\label{def:admissibility}
		A \emph{portfolio process} is an $(\mathcal{F}_t)$-predictable process $\pi:[0,T]\times\Omega\rightarrow\mathbb{R}$. 
		\begin{enumerate}[(i)]
			\item\label{DefAdmi} $\pi$ is called \emph{post-crash-admissible}, if it is non-negative\footnote{This means we rule out short sales of the risky asset.}, continuous and, if we do not have a purely Brownian model,  $l^{L}_{\max}\pi \leq 1$ and $$\mathbb{E}\left[\int_0^T\int_{[0,l^{L}_{\max}]}|\log(1-\pi_t l)|\vartheta(dl)dt\right]<\infty\footnote{The integrability condition is required to make the stochastic integral process $t\mapsto\int_{[0,l^{L}_{\max}]}\log(1-\pi_t l)\tilde{\nu}(ds,dl)$ a true martingale (see \cite[Section 4.3.2]{applebaum}).}.$$
			\item\label{DefAdmii} $\pi$ is called \emph{pre-crash-admissible}, if it is post-crash-admissible and satisfies in addition $l^{WOC}\pi\leq 1$ everywhere.
		\end{enumerate}
		Denote by $\mathcal{A}$ the set of all pre-crash-admissible and by $\bar{\mathcal{A}}$ the set of all post-crash-admissible portfolio processes
	\end{Def}
	
	\begin{Rem}
		Note that condition \eqref{DefAdmii} implies that $\pi$ is always bounded. Condition \eqref{DefAdmi} implies boundedness in the model including jumps.
		
		%The assumption for continuity of portfolio processes is important for the methods of this %paper. While the BSDE theory part later on seems to be apt for treating also jump models, the %procedures from Section \ref{sec:pre} seem to work with continuous processes $\pi$ only (and %naturally encourage research to find a workaround). Continuous portfolio strategies, in turn, %can only be usefully applied when the processes $\lambda$ and $\sigma$ are from the Brownian %filtration only, as assumed in Assumption %\eqref{assumption:stochCoeff:CoefficientsContinuous} (see also Remark %\ref{Rem:JumpsExcludeRemark}). A detailed explanation for the necessity of continuous %portfolio processes is given in Remark \ref{Rem:Continuity} below.
		\end{Rem}
	
	We consider the problem of an investor to choose a pair $(\pi,\bar{\pi})\in\mathcal{A}\times\bar{\mathcal{A}}$ consisting of a pre-crash and a post-crash strategy %\footnote{In some formulations the post-crash process $\bar{\pi}$ is allowed to be chosen contingent on the actual crash scenario $\tau$ being realized. However, the optimal post-crash choice turns out to be independent of the crash scenario $\tau$, so it is without loss of generality to consider the simpler choice set with a common post-crash process.} 
	in order to maximize the utility of terminal wealth in the worst-case crash scenario, namely to maximize
	\begin{align}
	\inf_{\tau\in\Theta}{\mathbb{E}\left[\log{X_T^{(\pi,\bar{\pi}),\tau}}\right]}. \label{eq:P} \tag{P}
	\end{align}
	Here, $X^{(\pi,\bar{\pi}),\tau}$ denotes the wealth process of the investor, if he follows $\pi$ prior and up to the crash, $\bar{\pi}$ after the crash and the crash happens at time $\tau$, that is $X^{(\pi,\bar{\pi}),\tau}$ is the solution $X$ to
	\begin{align*}
		\frac{dX_t}{X_{t-}} &= (r_t + \pi_t\lambda_t)dt + \pi_t\sigma_t dW_t - \int_{[0,l^{L}_{\max}]}\pi_tl\nu(dt,dl) &&\text{on }[0,\tau\wedge T)\,,\\
		X_{\tau\wedge T} &= (1- l^{WOC} \pi_{\tau\wedge T}) X_{\tau\wedge T-}\,,		\\
		\frac{dX_t}{X_{t-}} &= (r_t + \bar{\pi}_t\lambda_t)dt + \bar{\pi}_t\sigma_t dW_t -\int_{[0,l^{L}_{\max}]}\bar{\pi}_tl \nu(dt,dl)&&\text{on }(\tau\wedge T,T]\,,
	\end{align*}
	for some initial wealth $X_0=x>0$. 
	
	Note that the SDEs above are driven by the Brownian motion $W$ with coefficients that are measurable w.r.t.~a larger  filtration than the one generated by $W$ only. However, the usual It\^o's formula (and its various extensions) can still be applied as the integrals w.r.t.~$dW$ can always be written as e.g.~$$\pi_t\sigma_t dW_t=\pi_t\sigma_t\begin{pmatrix}\rho & \sqrt{1-\rho^2}\end{pmatrix}\begin{pmatrix}d\hat{W}_t\\d\tilde{W}_t\end{pmatrix}.$$ 
	
	The solution to the above SDE can then be given explicitly:

	\begin{Lem}\label{Lem:explicitObjectiveRepresentation}
		Let $(\pi,\bar{\pi})\in \mathcal{A}\times\bar{\mathcal{A}}$ be an admissible choice and $\tau\in\Theta$. Then a unique solution $X^{(\pi,\bar{\pi}),\tau}$ to the above forward SDE exists. This solution is strictly positive and $\log{X^{(\pi,\bar{\pi}),\tau}_t}$ is quasi-integrable for each $t$. Furthermore,
\begin{align*}\mathbb{E}\left[\log{X_T^{(\pi,\bar{\pi}),\tau}}\right] = &\log{x} + \mathbb{E}\left[\log\left(1 - 1_{\{\tau\leq T\}}\pi_{\tau}l^{WOC}\right) + \int_0^{\tau\wedge T}{\Phi_t(\pi_t)dt}\right] \\&+ \mathbb{E}\left[\int_{\tau\wedge T}^{T}{\Phi_t(\bar{\pi}_t)dt}\right].
\end{align*}
	\end{Lem}
	\begin{proof}
		Define first the auxiliary portfolio process $\tilde{\pi}$ as follows
		$$\tilde{\pi}_t(\omega) = \begin{cases}\pi_t(\omega),&t<\tau(\omega) \\\bar{\pi}_t(\omega),&t\geq \tau(\omega) \end{cases}$$
		and consider the crash-free SDE
		$$\frac{d\tilde{X}_t}{\tilde{X}_{t-}} = (r_t + \tilde{\pi}_t\lambda_t)dt + \tilde{\pi}_t\sigma_t dW_t -\int_{[0,l^{L}_{\max}]}\tilde{\pi}_tl\nu(dt,dl),\qquad \tilde{X}_0=x.$$
		Obviously, for any solution $\tilde{X}$ to this SDE, $X_t := \left(1 - 1_{\{\tau \leq t\}} \pi_{\tau} l^{WOC}\right)\tilde{X}_t$ solves the above SDE containing a crash and vice versa (meaning for any solution $X$ to the SDE with crash we can construct exactly one solution $\tilde{X}$ to the crash-free SDE). Now, the crash-free SDE is a linear SDE with integrable drift and square-integrable diffusion coefficient and integrable L\'evy measure with a unique (strong) solution $\tilde{X}$, given by
		\begin{align*}
			\tilde{X}_t &=  x\exp\bigg(\int_0^{t}{\left(r_s + \tilde{\pi}_s\lambda_s - \frac{1}{2}\tilde{\pi}_s^2\sigma^2_s +\int_{[0,l^{L}_{\max}]}\log(1-\tilde{\pi}_sl)\vartheta(dl)\right)ds} \\
			&\quad+ \int_{0}^{t}{\tilde{\pi}_s\sigma_sdW_s} +\int_{(0,t]\times[0,l^{L}_{\max}]}\log(1-\tilde{\pi}_sl)\tilde{\nu}(ds,dl)\bigg)\\
			&=  x\exp\bigg(\int_0^{\tau\wedge t}{\Phi_s(\pi_s)ds} + \int_{\tau\wedge t}^{t}{\Phi_s(\bar{\pi}_s)ds} + \int_0^{\tau\wedge t}{\pi_s\sigma_sdW_s} + \int_{\tau\wedge t}^{t}{\bar{\pi}_s\sigma_sdW_s}\\
			&\quad+\int_{(0,\tau\wedge t]\times[0,l^{L}_{\max}]}\log(1-{\pi}_sl)\tilde{\nu}(ds,dl)+\int_{(\tau\wedge t,t]\times[0,l^{L}_{\max}]}\log(1-\bar{\pi}_sl)\tilde{\nu}(ds,dl)\bigg).
		\end{align*}
		Here, the second line immediately follows from definitions of $\Phi$ and $\tilde{\pi}$.
		
		Clearly $\tilde{X}$ is strictly positive and, as $\pi_{\tau\wedge T}\leq \frac{1}{l^{WOC}}$ by pre-crash admissibility, so is
		$$X_t =  \left(1 - 1_{\{\tau \leq t\}} \pi_{\tau} l^{WOC}\right)\tilde{X}_t.$$
		
		It remains to show quasi-integrability of $\log{X_t}$ and the asserted representation of $\mathbb{E}\left[\log{X_T}\right]$. We have
		\begin{align*}
			\log{X_t} &= \log{x} + \log{\left(1 - 1_{\{\tau \leq t\}} \pi_{\tau} l^{WOC}\right)} + \int_0^{\tau\wedge t}{\Phi_s(\pi_s)ds} + \int_{\tau\wedge t}^{t}{\Phi_s(\bar{\pi}_s)ds}\\
			&\qquad + \int_0^{\tau\wedge t}{\pi_s\sigma_sdW_s} + \int_{\tau\wedge t}^{t}{\bar{\pi}_s\sigma_sdW_s}\\
			&\quad+\int_{(0,\tau\wedge t]\times[0,l^{L}_{\max}]}\log(1-{\pi}_sl)\tilde{\nu}(ds,dl)+\int_{(\tau\wedge t,t]\times[0,l^{L}_{\max}]}\log(1-\bar{\pi}_sl)\tilde{\nu}(ds,dl)
		\end{align*}
		By assumptions \eqref{assumption:stochCoeff:CoefficientsContinuous}, \eqref{assumption:stochCoeff:integrability} and boundedness of $\pi$ and $\bar{\pi}$, the stochastic integrals are martingales and have expectation $0$. For the same reason, $\int_0^{\tau\wedge t}{\Phi_s(\pi_s)ds}$ and $\int_{\tau\wedge t}^{t}{\Phi_s(\bar{\pi}_s)ds}$ are integrable random variables. Finally, $\log{\left(1 - 1_{\{\tau \leq t\}} \pi_{\tau} l^{WOC}\right)}$ might not be an integrable random variable, but due to $\pi\geq0$ this term is non-positive and thus trivially quasi-integrable. Hence, also $\log{X_t}$ is quasi-integrable. The asserted representation of $\mathbb{E}\left[\log{X_T}\right]$ is now a trivial conclusion.
	\end{proof}

\section{Solution to the Post-Crash Problem}\label{sec:post}	
	As is common in the worst-case optimal investment literature, the above problem can be solved by first considering for each crash scenario $\tau$ the post-crash problem starting at time $\tau$, which is a classical portfolio optimization problem, compare e.g.~Korn and Wilmott \cite{KornWilmott2002}, Seifried \cite{Seifried2010}. Using the explicit representation of the objective from Lemma \ref{Lem:explicitObjectiveRepresentation}, the following result is immediate here.
	
	\begin{Prop}[Solution of the post-crash problem]\label{prop: argmax}
		\ \\
			With \eqref{assumption:stochCoeff:integrability}, in the L\'evy model the optimal post-crash portfolio process is given by $\pi_t^M := \mathrm{argmax}_{\left[0,\frac{1}{l^{L}_{\max}}\right]}\Phi_t$.
	\end{Prop}
	\begin{proof}
			We consider two possible cases:\\
			{\bf Case 1:} \begin{align}\label{eq:integrabilityLevyCase1}
				\lim_{y\to\frac{1}{l^{L}_{\max}}}\int_{[0,l^{L}_{\max}]}\log(1-yl)\vartheta(dl)=\int_{[0,l^{L}_{\max}]}\log\left(1-\frac{l}{l^{L}_{\max}}\right)\vartheta(dl)= -\infty.
			\end{align}
			By assumptions \eqref{assumption:stochCoeff:CoefficientsContinuous} and since for all $t\geq 0$, $\Phi_t(y)\leq r_t+\lambda_t^+y$, and by \eqref{eq:integrabilityLevyCase1}, possible maxima of $\Phi_t$ are attained in $\left[0,\frac{1}{l^{L}_{\max}}\right)$.
			
			Hence on an interval, $[0,y'_t]$, $y'_t<\frac{1}{l^{L}_{\max}}$,  where the maxima of $\Phi_t$ are contained, we may differentiate $\Phi_t(y)$ twice in direction $y$ to see that $\Phi_t$ is strictly concave, thus $\bar{\pi}:=\mathrm{argmax}_{\left[0,\frac{1}{l^{L}_{\max}}\right)}\Phi_t$ is well defined and
			%\texttt{ \eqref{assumption:stochCoeff:MertonBounded}  $\bar{\pi} := \frac{\lambda}{\sigma\sigma^T}\vee 0$ 
				constitutes a continuous\footnote{See Stokey et.al., Recursive Methods in Economic Dynamics, Thm. 3.6}, bounded and adapted process. To show $\bar{\pi}\in\bar{\mathcal{A}}$, note that $|\log(1-yl)|\leq \frac{ly}{1-yl}$ for $y\in [0,y'_t]$. On the same interval, we know that 
				$$\partial_y\Phi_t(y)=\lambda_t-\sigma^2_t y-\int_{[0,l^{L}_{\max}]}\frac{l}{1-yl}\vartheta(dl).$$
				As this derivative is smaller or equal to zero for $y=\bar{\pi}_t$ (it equals zero whenever $\bar{\pi}_t>0$), we infer
				\begin{align*}
					\int_{[0,l^{L}_{\max}]}|\log(1-\bar{\pi}_tl)|\vartheta(dl)\leq \int_{[0,l^{L}_{\max}]}\frac{l\bar{\pi}_t}{1-\bar{\pi}_tl}\vartheta(dl)=\bar{\pi}_t\left(\lambda_t-\sigma^2_t\bar{\pi}_t\right).
				\end{align*}
				Thus, \eqref{assumption:stochCoeff:integrability} and the boundedness of $\bar{\pi}$ implies $$\mathbb{E}\left[\int_0^T\int_{[0,l^{L}_{\max}]}|\log(1-\bar{\pi}_tl)|\vartheta(dl)dt\right]<\infty$$ and herewith $\bar{\pi}\in\bar{\mathcal{A}}$.\\
				{\bf Case 2:} 
				\begin{align}\label{eq:integrabilityLevyCase2}
					\lim_{y\to\frac{1}{l^{L}_{\max}}}\int_{[0,l^{L}_{\max}]}\log(1-yl)\vartheta(dl)=\int_{[0,l^{L}_{\max}]}\log\left(1-\frac{l}{l^{L}_{\max}}\right)\vartheta(dl)> -\infty.
				\end{align}
				This case is simpler to treat since the monotone limit \eqref{eq:integrabilityLevyCase2} is finite, which yields that $\int_{[0,l^{L}_{\max}]}\log(1-yl)\vartheta(dl)>-\infty$ for all $y\in \left[0,\frac{1}{l^{L}_{\max}}\right]$. Again we get that $\Phi_t$ is a differentiable, strictly convex function, but now $\bar{\pi_t}=\mathrm{argmax}_{\left[0,\frac{1}{l^{L}_{\max}}\right]}\Phi_t$ may attain the value $\frac{1}{l^{L}_{\max}}$. The finiteness of \eqref{eq:integrabilityLevyCase2} now shows that $\bar{\pi}\in \bar{\mathcal{A}}$.

				In any of the two cases, for each $(t,\omega)$, $\bar{\pi}_t(\omega)$ is the unique global maximizer in $\left[0,\frac{1}{l^{L}_{\max}}\right]$ of the function $\Phi_t(\cdot)(\omega)$, so for any stopping time $\tau$ and any second $\bar{\pi}'\in\bar{\mathcal{A}}$ we have
				$$\int_{\tau\wedge T}^{T}{\Phi_t(\bar{\pi}_t)dt}\geq \int_{\tau\wedge T}^{T}{\Phi_t(\bar{\pi}_t')dt}.$$
				As the first and second term in the objective representation of Lemma \ref{Lem:explicitObjectiveRepresentation} do not depend on the choice of the post-crash portfolio process, this proves optimality of $\bar{\pi}$.
		\end{proof}
		
		%For future reference we denote the optimal post-crash strategy $\mathrm{argmax}_{\left[0,\frac{1}{l^{L}_{\max}}\right]}\Phi_t$ in the following by $\pi_t^M$ (Merton strategy). 
		
		\begin{Cor}
				In the model without jumps, $\pi^M_t$ is given by the classical Merton strategy $\frac{\lambda_t}{\sigma_t^2}\vee 0$. %which is admissible because of the conditions \eqref{assumption:stochCoeff:CoefficientsContinuous} and  \eqref{assumption:stochCoeff:MertonBounded}.}
		\end{Cor}
				
		We moreover need the following useful property of the post-crash optimal strategy later when proving optimality results; the proof can be found in Appendix~\ref{app:post}.
		
			\begin{Prop}\label{prop: psicont}
				The post-crash optimal strategy $\mathrm{argmax}_{\left[0,\frac{1}{l^{L}_{\max}}\right]}\Phi_t$ can be expressed by a two-variable function $\psi$, $\pi^M=\psi(\lambda,\sigma)$, where $\psi$ is a continuous function. Additionally, when $\tilde\sigma:=\inf\{|\sigma_t(\omega)|:(t,\omega)\in [0,T]\times\Omega \}>0$, we get Lipschitz continuity together with the following relations:
				\begin{align*}
					|\psi(\lambda,\sigma)-\psi(\lambda',\sigma)|\leq \frac{1}{\tilde{\sigma}^2}|\lambda-\lambda'|,
				\end{align*}
				and
				\begin{align*}
					|\psi(\lambda,\sigma)-\psi(\lambda,\sigma')|\leq \frac{2}{l^{L}_{\max}\cdot|\tilde\sigma|}|\sigma-\sigma'|.
				\end{align*}
			\end{Prop}

	\begin{Def}\label{def:appropriate}
					We say that such a market price of risk $\lambda$ is $\psi$-appropriate, if $\psi$ is the function representing $\pi^M$ by $\pi^M=\psi(\lambda,\sigma)$.  
					
					For example, in the model without jumps, $r$, $\lambda$ and $\sigma$ are said to constitute a \emph{linear market price of risk model}, if for some $\alpha\geq 0$, $\psi(\lambda,\sigma)=\frac{\lambda}{\sigma^2}=\alpha$.

				\end{Def}
		Such $\psi$-appropriate market prices of risk are computed e.g.~in Subsection \ref{subsec:jumpModel}. 
		\begin{Rem}[Power utility with stochastic coefficients]\hfill\\
				In the case of stochastic coefficients $\lambda, \sigma$ we can not directly apply the change of measure device procedure in order to solve the Merton problem if we replace the logarithm by the power function $u=x\mapsto \frac{x^{\rho}}{\rho}$ for $\rho < 1, \rho \neq 0$, as it was done in \cite{Seifried2010} having deterministic market coefficients. In particular:
				$$u(X_T^{\bar{\pi}})=u(X_\tau^{\bar{\pi}})\exp\left(\rho\int_\tau^T\bar{\Phi}(\bar{\pi}_t)dt\right)M_T(\bar{\pi})$$
				and
				\begin{align*}
					&\mathbb{E}\left[u(X_T^{\bar{\pi}})\right]\leq \mathbb{E}\left[u(X_\tau^{\bar{\pi}^M})\exp\left(\rho\int_\tau^T\bar{\Phi}(\bar{\pi}^M_t)dt\right)M_\tau(\bar{\pi})\right]\\
					&\stackrel{\text{not clear}}{=} \mathbb{E}\left[u(X_\tau^{\bar{\pi}^M})\exp\left(\rho\int_\tau^T\bar{\Phi}(\bar{\pi}^M_t)dt\right)M_T(\bar{\pi}^M)\right],
				\end{align*}
				where the equality does not have to hold as $\int_\tau^T\bar{\Phi}(\bar{\pi}^M_t)dt$ is not $\mathcal{F}_\tau$-measurable. 
		\end{Rem}
		
		The following assumption \eqref{assumption:stochCoeff:MertonBounded} for the Merton strategy $\frac{\lambda}{\sigma^2}\vee 0$ is necessary for the model without jumps, as we need our strategies to be admissible. The second assumption is used later on in examples.
		\begin{align}
		&\frac{\lambda}{\sigma^2}\mbox{ is bounded}, \tag{C1} \label{assumption:stochCoeff:MertonBounded}\\
		&\pi^M \mbox{ is constant}\,.  \tag{C2}  \label{assumption:stochCoeff:MertonConstant}
	\end{align}

		\section{Solution to the Pre-Crash Problem}\label{sec:pre}
		
		It is now left to find the optimal pre-crash strategy $\pi$. This is the main concern of the present paper. 
		
		Given the known solution to the post-crash problem, we simplify the worst-case problem as follows. First, ignore the constant term $\log{x}$ in the objective. As long as $x>0$ its value does not matter for the optimal choice. Second, rewrite the remaining objective as
			$$\mathbb{E}\left[\log\left(1 - \pi_{\tau}l^{WOC}\right) + \int_0^{\tau\wedge T}{\left(\Phi_t(\pi_t)-\Phi_t(\pi_t^M)\right)dt} \right] + \mathbb{E}\left[ \int_0^{T}{\Phi_t(\pi_t^M)dt}\right],$$
			which has the advantage, that the third summand does neither depend on $\tau$, nor on $\pi$ and can therefore be ignored, while the argument of the expectation in the second summand is $\mathcal{F}_{\tau}$-measurable
		
		%\begin{enumerate}
		%	\item convention:\\
		%	to get rid of indicator terms $1_{\{\tau\leq T\}}$, extend pre-crash portfolio processes from the time domain $[0,T]$ to $[0,T]\cup\{\infty\}$ not by our usual assignment $\pi_{\infty}:=\pi_T$, but instead by $\pi_{\infty}:=0$.
		%	\item reformulation/transformation:\\
			
		%\end{enumerate}
		
		We arrive at the following problem:
		
		\begin{Prob}[Pre-Crash Problem]\label{Prob:WC:PreCrashWithStochCoefficients}
			Choose a portfolio strategy $\pi\in \mathcal{A}$ to maximize
			\begin{align}\label{eq:Ppre}
		 \inf_{\tau\in\Theta}{\mathbb{E}\left[\log{(1-\pi_\tau l^{WOC})} + \int_{0}^{\tau\wedge T}{\left(\Phi_s(\pi_s) - \Phi_s(\pi_s^M)\right) ds}\right]} 	\tag{$P_{Pre}$}
			\end{align}
		\end{Prob}
		
		\medskip
		
		In the following we define the process
		\begin{align*}
		Z_t^{\pi}  &:= \log{(1-\pi_t l^{WOC})} + \int_{0}^{t\wedge T}{\left(\Phi_s(\pi_s) - \Phi_s(\bar{\pi}_s^M)\right) ds},\quad t\in [0,T],\\
		Z_\infty^{\pi}&:=\int_{0}^{T}{\left(\Phi_s(\pi_s) - \Phi_s(\bar{\pi}_s^M)\right) ds}.
		\end{align*}

		Furthermore, define the utility crash exposure $\Upsilon^\pi$ of strategy $\pi\in\mathcal{A}$ by
		$$\Upsilon_t^\pi := -\log{(1-\pi_t l^{WOC})}.$$

		\begin{Rem}
		The seminal work \cite{Seifried2010} also defines a crash exposure process, but with respect to wealth, which is different from the exposure w.r.t.~utility which we introduce here.
		\end{Rem}
		
		In order to properly characterize (worst-case) optimality of strategies, we introduce the following notion:  	
		
		\begin{Def}[Worst-case dominance]
		A strategy $\pi\in \mathcal{A}$ is said to worst-case dominate a strategy $\pi'\in\mathcal{A}$ and, equivalently, $\pi'$ is said to be worst-case dominated by $\pi$, if for every stopping time $\tau\in\Theta$, there is a stopping time $\theta\in\Theta$, such that
		$$\mathbb{E}[\tau^\pi]\geq \mathbb{E}[\theta^{\pi'}]$$
		Write in this case $\pi\succ \pi'$.
		\end{Def}
		
		We record the following straight-forward result:
		
		\begin{Lem}\label{Lem:WorstCaseDominanceImpliesOptimality}
		If a strategy $\pi^\ast\in\mathcal{A}$ satisfies $\pi^\ast\succ \pi$ for all $\pi\in\mathcal{A}$, then  $\pi^\ast\in\mathcal{A}$ solves Problem \eqref{eq:Ppre}. In addition, if the infimum in the objective \eqref{eq:Ppre}is always attained, i.e.~if there is always a worst-case scenario, then also the converse holds.
		\end{Lem}
	  \begin{proof}
		Suppose, $\pi^\ast$ worst-case dominates any other strategy $\pi$. Then for any stopping time $\tau\in\Theta$ there is a $\theta\in\Theta$ such that
		$$\mathbb{E}[Z_{\tau}^{\pi^\ast}]\geq \mathbb{E} [Z_{\theta}^{\pi}]\geq \inf_{\theta\in\Theta}{ \mathbb{E}[\theta^{\pi}]}.$$
		Taking the infimum over all $\tau\in\Theta$ and then the supremum over all $\pi\in \mathcal{A}$ yields the conclusion.
		
		Conversely, if $\pi^\ast$ solves \eqref{eq:Ppre} and $\pi\in\mathcal{A}$ is arbitrary, then for any stopping time $\tau\in\Theta$
		$$\mathbb{E}[Z_{\tau}^{\pi^\ast}]\geq \inf_{\theta\in\Theta}{\mathbb{E}[Z_{\theta}^{\pi^\ast}]} \geq \inf_{\theta\in\Theta}{\mathbb{E}[Z_{\theta}^{\pi}]}.$$
		If the infimum on the right-hand side is attained, this implies $\pi^\ast\succ \pi$.
		\end{proof}
		
		Note that $\succ$ is not a preorder.
		
		\subsection{Super- and Subindifference Strategies}
	
		%
		%\medskip
		
		In this section we extend the definition of an indifference strategy to the terms superindifference strategy and subindifference strategy and derive several results to bound the worst-case optimal strategy - if it exists - from below and above.
		%\medskip
		
		%{\color{red} for approximation arguments below it makes sense to start out already here with the set $\tilde{A}$ and define indifference strategies generally without requiring continuity properties of paths. The second proposition in the BSDE section always gives us continuity for free...}
		
		We define the notion of (sub-/super-)indifference strategies for a more general class of processes than the one of all admissible portfolio processes $\mathcal{A}$. We will see later, that at least indifference strategies are automatically contained in $\mathcal{A}$. Let $\tilde{\mathcal{A}}$  be the set of all progressively measurable bounded processes $\pi:[0,T]\times\Omega\rightarrow \mathbb{R}$ that satisfy the pre-crash admissibility conditions of Definition \ref{def:admissibility} except for the requirement of continuity. 
		
		%I  The adjectives ``superindifferent" and ``subindifferent" do not really make any sense, but this terminology is motivated by the sub-/supermartingale properties it requires. If you have a better term, I am happy to take any suggestion.
		
		\begin{Def}[Super-/Subindifference Strategy]\label{Def:IndifferenceStrategy}
			Let $\underline{\varrho}$, $\overline{\varrho}\in \Theta$ be two stopping times with $\underline{\varrho}\leq \overline{\varrho}$ and $\pi\in \tilde{\mathcal{A}}$ a portfolio process.
			\begin{itemize}
				\item $\pi$ is called a superindifference strategy on $[\underline{\varrho},\overline{\varrho}]$, if $Z^{\pi}$ is a supermartingale on $[\underline{\varrho},\overline{\varrho}]$.
				\item $\pi$ is called a subindifference strategy on $[\underline{\varrho},\overline{\varrho}]$, if $Z^{\pi}$ is a submartingale on $[\underline{\varrho},\overline{\varrho}]$.
				\item $\pi$ is called an indifference strategy on $[\underline{\varrho},\overline{\varrho}]$, if it is both a super- and subindifference strategy on $[\underline{\varrho},\overline{\varrho}]$.
			\end{itemize}
			The reference to $[\underline{\varrho},\overline{\varrho}]$ is omitted, if it coincides with $[0,T]\cup\{\infty\}$.
		\end{Def}
		
		%\paragraph{\textit{Interpretation:}}
		In the controller-vs-stopper-game setting of \cite{Seifried2010} we can interpret these strategies as follows:
		\begin{itemize}
			\item A subindifference strategy is a strategy such that at any given time the market/stopper's best response is to stop the continuation game starting at that time immediately.
			\item A superindifference strategy is a strategy such that the market/stopper's best response is to wait forever and never stop the game early.
			\item An indifference strategy is a strategy such that the market/stopper is at any point in time indifferent between stopping and waiting.
		\end{itemize}
		
\begin{Rem}\label{rem:indiffNoJump}
For a superindifference strategy, the process $Z^{\pi}$ needs to be a supermartingale on $\{T,\infty\}$ as well. This represents the indifference to a crash happening at the last moment, or not at all. Since both $Z^\pi_T$ and $Z^\pi_\infty$ are $\mathcal{F}_T$ measurable, it follows that ${\pi_T}=0$. See also \cite[4.2. Indifference at $\infty$]{Seifried2010}.
\end{Rem}

		In what follows we also need to be able to concatenate strategies:
		
		\begin{Def}
			Let $\pi_1,\pi_2\in\mathcal{A}$ and $\varrho\in\Theta$ a  stopping time. Then the process $\pi_1|_{\varrho}\pi_2$ that switches at $\varrho$ from  is defined by
			$$(\pi_1|_{\varrho}\pi_2)_t(\omega)=\begin{cases}\pi_{1,t}(\omega),& t<\varrho(\omega)\\ \pi_{2,t}(\omega),& t\geq \varrho(\omega)\end{cases}.$$
			
			Use for the following two cases a special notation:
			\begin{itemize}
				\item We write $\pi_1\uparrow\pi_2$ instead of $\pi_1|_{\varrho}\pi_2$, if $\varrho=\inf{\{t\geq 0\mid \pi_{1t} >  \pi_{2t}\}}$.
				\item We write $\pi_1\downarrow\pi_2$ instead of $\pi_1|_{\varrho}\pi_2$, if $\varrho=\inf{\{t\geq 0\mid \pi_{1t} < \pi_{2t}\}}$.
			\end{itemize}
		\end{Def}
		
		\begin{Rem}
			Obviously, it is again $\pi_1|_{\varrho}\pi_2\in\mathcal{A}$.
		\end{Rem}

		We will also make use of the following lemma throughout the remainder of the section.
		
		\begin{Lem}\label{Lem:concatenationLemma}
			Let $\pi_1,\pi_2\in\mathcal{A}$, $\varrho\in\Theta$ and $\pi=\pi_1|_{\varrho}\pi_2$.
			\begin{enumerate}[(a)]
				\item If $\pi_1$ is a (super-/sub-)indifference strategy on $[0,\varrho)$, then $\pi$ is, too.
				\item If $\pi_2$ is a (super-/sub-)indifference strategy on $[\varrho,T]$, then so is $\pi$.
			\end{enumerate}
		\end{Lem}
		\begin{proof}
			Part (a) follows directly from $Z^{\pi_1}|_{[0,\varrho)}=Z^{\pi}|_{[0,\varrho)}$.
			
			For the proof of part (b) let $\tau\geq \varrho$ be an arbitrary stopping time. Then
			\begin{align*}
				Z^\pi_\tau &= -\Upsilon_{\tau}^\pi + \int_0^{T\wedge\tau}{\left(\Phi_s(\pi_s) - \Phi_s(\bar{\pi}_s^M)\right) ds}\\
				&= -\Upsilon_{\tau}^{\pi_2} + \int_0^{\varrho}{\left(\Phi_s(\pi_{1,s}) - \Phi_s(\bar{\pi}_s^M)\right) ds} + \int_\varrho^{T\wedge\tau}{\left(\Phi_s(\pi_{2,s}) - \Phi_s(\bar{\pi}_s^M)\right) ds}\\
				&= -\Upsilon_{\tau}^{\pi_2} + \int_0^{\tau\wedge T}{\left(\Phi_s(\pi_{2,s}) - \Phi_s(\bar{\pi}_s^M)\right) ds} + \int_0^{\varrho}{\left(\Phi_s(\pi_{1,s}) - \Phi_s(\pi_{2,s})\right) ds}\\
				&= Z^{\pi_2}_\tau + \int_0^{\varrho}{\left(\Phi_s(\pi_{1,s}) - \Phi_s(\pi_{2,s})\right) ds}
			\end{align*}
			Since the term $\int_0^{\varrho}{\left(\Phi_s(\pi_{1,s}) - \Phi_s(\pi_{2,s})\right) ds}$ is $\mathcal{F}_{\varrho}$-measurable, this shows that $Z^\pi$ is a (sub-/super-)martingale on $[\varrho,T]$, if and only if $Z^{\pi_2}$ is. %({\color{blue} the integral term is always in $\mathcal{L}^1(P)$ by assumptions on an admissible process(??)} \textcolor{red}{yes exactly - will be fixed in the next version. That is a standard requirement already mentioned in the \cite{Kraft2005} paper}.)

		\end{proof}

		%\section{Optimality Bounds}\label{sec:OptimalityBounds}

		\subsection{The Subindifference Frontier}
		The following is an enhancement of the classical indifference frontier result of \cite{Seifried2010} for the log utility case with stochastic market coefficients

		\begin{Prop}[Subindifference Frontier]\label{Prop:subindifferenceFrontier}
			Let $\pi$ be an arbitrary admissible portfolio process and $\hat{\pi}$ a continuous subindifference strategy. Then $\pi\uparrow\hat{\pi}$ worst-case dominates~$\pi$.
		\end{Prop}
		\begin{proof}
			Let $\varrho$ be as in the definition of $\pi\uparrow\hat{\pi}=:\tilde{\pi}$. Then $\varrho$ is a stopping time by the debut theorem and $\tilde{\pi}\in\mathcal{A}$. Let now $\tau\in\Theta$ be an arbitrary stopping time and $\tilde{\Upsilon}$ the utility crash exposure process of $\tilde{\pi}$, $\Upsilon$ the one of $\pi$ and $\hat{\Upsilon}$ the one of $\hat{\pi}$. 
			
			By Lemma \ref{Lem:concatenationLemma}, $\tilde{\pi}$ is a subindifference strategy on $[\varrho,T]$ and thus $Z^{\tilde{\pi}}$ a submartingale in this interval. This implies (because of $\left\{\varrho < \tau\right\} =\allowbreak \Omega\setminus \{\tau \leq \varrho\}\in\mathcal{F}_{\varrho}$)
			\begin{align}\label{Lem:indifferenceFrontier:proof:eq1}
				\mathbb{E}\left[Z_{\tau}^{\tilde{\pi}}\right] &= \mathbb{E}\left[1_{\{\varrho \geq \tau\}}Z_{\tau}^{\tilde{\pi}}\right] + \mathbb{E}\left[1_{\{\varrho < \tau\}}Z_{\tau}^{\tilde{\pi}}\right]\notag\\
				&= \mathbb{E}\left[1_{\{\varrho \geq \tau\}}Z_{\tau\wedge \varrho}^{\tilde{\pi}}\right] + \mathbb{E}\left[1_{\{\varrho < \tau\}}Z_{\tau\vee \varrho}^{\tilde{\pi}}\right]\notag\\
				&= \mathbb{E}\left[1_{\{\varrho \geq \tau\}}Z_{\tau\wedge \varrho}^{\tilde{\pi}}\right] + \mathbb{E}\left[1_{\{\varrho < \tau\}}\mathbb{E}[Z_{\tau\vee \varrho}^{\tilde{\pi}}\mid \mathcal{F}_{\varrho}]\right]\notag\\
				&\geq \mathbb{E}\left[1_{\{\varrho \geq \tau\}}Z_{\tau\wedge \varrho}^{\tilde{\pi}}\right] + \mathbb{E}\left[1_{\{\varrho < \tau\}}Z_{\varrho}^{\tilde{\pi}}\right] = \mathbb{E}\left[Z_{\tau\wedge \varrho}^{\tilde{\pi}}\right]
			\end{align}
			The last expression can be decomposed as follows
			\begin{align}\label{Lem:indifferenceFrontier:proof:eq2}
				\mathbb{E}\left[Z_{\tau\wedge \varrho}^{\tilde{\pi}}\right] &= \mathbb{E}\left[1_{\{\varrho > \tau\}}Z_{\tau\wedge \varrho}^{\tilde{\pi}}\right] + \mathbb{E}\left[1_{\{\varrho \leq \tau\}}Z_{\tau\wedge \varrho}^{\tilde{\pi}}\right]\notag\\
				&= \mathbb{E}\left[1_{\{\varrho > \tau\}}Z_{\tau\wedge \varrho}^{\pi}\right] + \mathbb{E}\left[1_{\{\varrho \leq \tau\}}Z_{\tau\wedge \varrho}^{\tilde{\pi}}\right]
			\end{align}
			
			Next, because $\pi$ and $\tilde{\pi}$ coincide on $[0,\varrho)$ and $\tilde{\Upsilon}_{\varrho} = \hat{\Upsilon}_{\varrho} \leq \Upsilon_{\varrho}$ due to continuity of $\pi$ and $\hat{\pi}$, we can estimate
			\begin{align*}
				Z_{\varrho}^{\tilde{\pi}} &=  - \tilde{\Upsilon}_{\varrho} + \int_0^{\varrho\wedge T}{(\Phi_s(\tilde{\pi}_s)-\bar{\Phi}_s(\bar{\pi}_s^M))ds}\\
				&= - \tilde{\Upsilon}_{\varrho} + \int_0^{\varrho\wedge T}{(\Phi_s(\pi_s)-\bar{\Phi}_s(\bar{\pi}_s^M))ds}\\
				&\geq  - \Upsilon_{\varrho} + \int_0^{\varrho\wedge T}{(\Phi_s(\pi_s)-\bar{\Phi}_s(\pi_s^M))ds} = Z_{\varrho}^{\pi}.
			\end{align*}
			Combining this result with equations (\ref{Lem:indifferenceFrontier:proof:eq1}) and (\ref{Lem:indifferenceFrontier:proof:eq2}) yields
			\begin{align*}
				\mathbb{E}\left[Z_{\tau}^{\tilde{\pi}}\right] &\geq \mathbb{E}\left[1_{\{\varrho > \tau\}}Z_{\tau\wedge \varrho}^{\pi}\right] + \mathbb{E}\left[1_{\{\varrho \leq \tau\}}Z_{\tau\wedge \varrho}^{\tilde{\pi}}\right]\\
				&\geq \mathbb{E}\left[1_{\{\varrho > \tau\}}Z_{\tau\wedge \varrho}^{\pi}\right] + \mathbb{E}\left[1_{\{\varrho \leq \tau\}}Z_{\tau\wedge \varrho}^{\pi}\right] = \mathbb{E}\left[Z_{\tau\wedge \varrho}^{\pi}\right]
			\end{align*}
			
		\end{proof}
		
\begin{Rem}\label{Rem:Continuity}
As suggested by our filtration, which involves a Poisson random measure, and hence jumps, one might be tempted to generalize the reasoning to predictable strategies that might jump as well. However, under such a filtration, whenever we assume that the process $Z^\pi_t=\log(1-\pi_tl^{WOC})+\int_0^{t\wedge T}(\Phi_s(\pi_s)-\Phi_s(\pi^M_s))ds$ is a martingale (i.e.~$\pi$ is an indifference strategy), the jumps of $Z^\pi$ emerge from the term $\log(1-\pi_tl^{WOC})$, that is, from the strategy $\pi$ itself. A martingale in such a filtration (we assume the usual conditions) can always be assumed as its right-continuous version. If $\pi$, as a strategy, is predictable, also $Z$ becomes predictable. But a predictable martingale in a filtration formed by a Brownian motion and a Poisson random measure must be left-continuous already.

%For a seemingly possible workaround, namely to keep the strategy $\pi$ at a left-continuous limit %process $\pi_t=h_{t-}$ for some c\`adl\`ag, adapted process $h$ and assuming that $t\mapsto Z_{t+}$ %is a martingale, the above proof is troublesome to conduct: we cannot apply the (sub-)martingale %property e.g.~in \eqref{Lem:indifferenceFrontier:proof:eq1} involving the stopping time $\tau$ (this %could only be done if $\tau$ is an $(\mathcal{F}_{t-})_{t\in [0,T]}$-stopping time, which it is not %by its very nature as unpredictable crash time).
\end{Rem}

		\subsection{A Superindifference Frontier Result}
		
		The subindifference frontier result from \cite{Seifried2010} presented above states that any subindifference strategy bounds the worst-case optimal strategy from above. A natural question is then to ask: Under which conditions can a superindifferent strategy bound the worst-case optimal strategy from below? If this was always the case, we could immediately conclude, that any indifference strategy is always worst-case optimal. However, the worst-case problem has a certain degree of asymmetry, such that we cannot expect a result as tight as the one for the subindifference frontier: In general there is a trade-off between high risk-adjusted expected returns absent a crash and good crash protection, but whereas the latter is strictly decreasing in $\pi_t$, the first goal is represented by the function $\Phi_t$, which is not everywhere strictly increasing in $\pi_t$, but contains a quadratic risk-penalty term. Increasing $\pi_t$ thus always worsens the crash protection, but only leads to expected utility gains absent a crash, if $\Phi_t'$ is positive, i.e. as long as $\pi_t\leq \pi^M_t$ where $\pi^M_t$ is the post-crash optimal strategy (maximum of $\Phi_t$). Thus, we can only expect a superindifference frontier for superindifference strategies $\pi^0$ with the additional property that $\pi^0\leq \pi^M$. Then indeed, the following holds.
		
		\begin{Prop}[Superindifference Frontier]\label{Prop:SuperindifferenceFrontier}
			Let $\pi^0$ be a continuous superindifference strategy such that $\pi^0\leq \pi^M$ and $\pi_T^0=0$. Then $\pi\vee \pi^0$ worst-case dominates~$\pi$.
		\end{Prop}
		\begin{proof}
			Let $\tilde{\pi}:=\pi\vee \pi^0$, $\tau\in\Theta$ be an arbitrary stopping time and define $\rho$ by
			$$\varrho(\omega) := \inf\{t\geq \tau\mid \tilde{\pi}_t\leq \pi_t\} = \inf\{t\geq \tau\mid \pi_t^0\leq \pi_t\},$$
			which is a stopping time. By continuity of both, $\pi$ and $\pi^0$, we have the inequality $\pi_{\rho}^0\leq \pi_{\rho}$%\footnote{While $\rho(\omega)=\infty$ for some $\omega$ is permitted here, this is not a problem due to our convention that $\pi_{\infty}=0$ for all portfolio processes $\pi$.}. 
			Thus, $\tilde{\pi}_{\varrho} = \pi_{\varrho}$ and therefore $\tilde{\Upsilon}_{\varrho} = \Upsilon_{\varrho}$. In addition, for each $(t,\omega)$ either $\tilde{\pi}_t(\omega)=\pi_t^0(\omega)\leq \pi^M_t(\omega)$ or $\tilde{\pi}_t(\omega)=\pi_t(\omega)$ holds. Because $\Phi_t$ is increasing on $[0,\pi^M_t]$ and $\pi\leq \pi\vee \pi^0 = \tilde{\pi}$ by definition, we have $\Phi_t(\pi_t)\leq \Phi_t(\tilde{\pi}_t)$ and thus
			$$\int_0^{\varrho\wedge T}{(\Phi_s(\pi_s)-\Phi_s(\bar{\pi}_s^M))ds} \leq \int_0^{\varrho\wedge T}{(\Phi_s(\tilde{\pi}_s)-\Phi_s(\bar{\pi}_s^M))ds}.$$
			This together with $\tilde{\Upsilon}_{\varrho} = \Upsilon_{\varrho}$ implies $Z^{\tilde{\pi}}_{\varrho}\geq Z^{\pi}_{\varrho}$. Using that $\tilde{\pi}=\pi^0$ on $[\tau,\varrho]$ and thus $Z^{\tilde{\pi}}$ is a supermartingale there, we can conclude
			$$\mathbb{E}\left[Z^{\tilde{\pi}}_{\tau}\right] \geq \mathbb{E}\left[\mathbb{E}\left[Z^{\tilde{\pi}}_{\varrho}\mid \mathcal{F}_{\tau}\right]\right] = \mathbb{E}\left[Z^{\tilde{\pi}}_{\varrho}\right] \geq \mathbb{E}\left[ Z^{\pi}_{\varrho}\right].$$
			Since $\tau$ was arbitrary, this shows that $\pi$ is worst-case dominated by $\tilde{\pi}$.
			
		\end{proof}

		\subsection{The Merton Bound}
		
		As described above, there is no trade-off between risk-return performance absent a crash and a low crash exposure above the post-crash optimal strategy $\pi^M$. In this case, both the crash exposure $\Upsilon$ and $\Phi(\pi)$ are strictly decreasing in $\pi$ and thus it is unambiguously better to (marginally) decrease $\pi$. By this reasoning, it can never be optimal to invest a higher share than $\pi^M$ into the risky asset. Indeed, the following result holds.
		
		\begin{Lem}[Merton bound]\label{Lem:MertonBound}
			Let $\pi\in\mathcal{A}$. Then $\pi\wedge \pi^M$ worst-case dominates $\pi$.
		\end{Lem}
		\begin{proof}
			Let $\tilde{\pi} := \pi\wedge \pi^M$ and $\Upsilon$, $\tilde{\Upsilon}$ be the exposure processes of $\pi$ and $\tilde{\pi}$, respectively. Obviously, $\tilde{\pi}\in\mathcal{A}$ and since $\tilde{\pi}\leq \pi$, we have $\tilde{\Upsilon}\leq \Upsilon$. In addition, for all $\omega\in\Omega$ and all $t\in[0,T]$ either $\tilde{\pi}_t(\omega)=\pi_t(\omega)$, then trivially $\Phi_t(\tilde{\pi}_t)(\omega)=\Phi_t(\pi_t)(\omega)$, or $\tilde{\pi}_t(\omega)<\pi_t(\omega)$, then $\tilde{\pi}_t(\omega)=\pi^M_t(\omega)$ and thus $\Phi_t(\tilde{\pi}_t)(\omega)=\Phi_t(\pi_t^M)(\omega)\geq \Phi_t(\pi_t)(\omega)$. Hence, $\Phi(\pi)\leq \Phi(\tilde{\pi})$ everywhere. Combining these two properties we have
			$$Z^{\tilde{\pi}}_t = -\tilde{\Upsilon}_t + \int_0^{t\wedge T}{\left(\Phi_s(\tilde{\pi}_s)-\Phi_s(\bar{\pi}^M)\right)ds} \geq -\Upsilon_t +  \int_0^{t\wedge T}{\left(\Phi_s(\pi_s)-\Phi_s(\bar{\pi}^M)\right)ds} = Z^\pi_t.$$
			This implies $Z^{\tilde{\pi}}_{\tau}\geq Z^\pi_{\tau}$ for all stopping times $\tau\in\Theta$, which is clearly a stronger property than $\tilde{\pi}\succ \pi$.
		\end{proof}

		\subsection{Worst-Case Optimality of Indifference Strategies}
		
		Next we prove the crucial optimality result for the worst-case problem with stochastic market coefficients. In particular, this optimality holds whenever the indifference strategy is dominated by the post-crash optimal strategy:
		
		\begin{Thm}\label{Prop:OptimalityPropForLargeMerton}
			Let $\hat{\pi}$ be a continuous indifference strategy and suppose that $\hat{\pi}_t\leq \pi^M_t$ $\mathbb{P}$-a.s.\ for all $t\in [0,T]$. Then $\hat{\pi}$ solves the pre-crash portfolio problem \eqref{eq:Ppre}.
		\end{Thm}
		\begin{proof}
			Let $\pi$ be an arbitrary admissible strategy. As a continuous indifference strategy, $\hat{\pi}$ is a continuous subindifference strategy and thus by the subindifference frontier result of Proposition \ref{Prop:subindifferenceFrontier}, $\pi\uparrow\hat{\pi} \succ \pi$. 
			
			On the other hand, $\hat{\pi}$ is also a continuous superindifference strategy with $\hat{\pi}_T=0$, (see Remark \ref{rem:indiffNoJump}) and by assumption $\hat{\pi}\leq \pi^M$. Since we have $\hat{\pi} \geq \pi\uparrow\hat{\pi}$ by definition of  $\pi\uparrow\hat{\pi}$, Proposition \ref{Prop:SuperindifferenceFrontier} implies $\hat{\pi} = \hat{\pi}\vee(\pi\uparrow\hat{\pi} ) \succ \pi\uparrow\hat{\pi} $.
			
			Combining these two arguments shows $\hat{\pi}\succ \pi$. Since $\pi$ was arbitrary, any admissible portfolio process is worst-case dominated by $\hat{\pi}$ and Lemma \ref{Lem:WorstCaseDominanceImpliesOptimality} implies the assertion.
		\end{proof}

		\section{A BSDE Characterisation of Indifference Strategies}\label{sec:indifferenceBSDE}

		In the previous section we have seen how (super-/sub-)indifference strategies can be useful to derive bounds for the worst-case optimal solution. In this section we discuss indifference strategies in more detail using a characterization in terms of backward stochastic differential equations (BSDEs). This is completely analogous to the ODE characterization of indifference strategies in the literature on worst-case optimization for constant market coefficients, cf.~for instance \cite{KornWilmott2002,KornMenkens2005,KornSeifried2009}.
		In what follows we use the following notations for BSDEs:
			Let $\bar{W}$ denote the vector of our independent driving Brownian motions, $$\bar{W}:=\begin{pmatrix}
			\hat{W}\\ \tilde{W}
			\end{pmatrix}.$$	
		
			Let  $\mathcal{S}^p$ denote the  space of all $(\mathcal{F}_t)$-progressively measurable and c\`adl\`ag processes  $\Upsilon\colon\Omega\times{[0,T]} \rightarrow \mathbb{R}$ such that
			\begin{align*}
				\left\|\Upsilon\right\|_{\mathcal{S}^p}:=\left\|\sup_{0\leq t\leq T} \left|\Upsilon_{t}\right|\right\|_p  <\infty.
			\end{align*}	
			 Let $\mathcal{D}$ be the space of all $(\mathcal{F}_t)$-progressively measurable and c\`adl\`ag processes  $\Upsilon\colon\Omega\times{[0,T]} \rightarrow \mathbb{R}$ such that
			\begin{align*}
				\sup_{\tau}\mathbb{E}\left[\left|\Upsilon_\tau\right|\right] <\infty,
			\end{align*}
			where the supremum is taken over all $(\mathcal{F}_t)$-stopping times $\tau$.
			We define $L^p(\bar{W}) $ as the space of all $(\mathcal{F}_t)$-progressively measurable processes $\sigma_{\Upsilon}\colon \Omega\times{[0,T]}\rightarrow \mathbb{R}^{1\times 2}$  such that
			\begin{align*}
				\left\|\sigma_{\Upsilon}\right\|_{L^p(\bar{W}) }:=\mathbb{E}\left[\left(\int_0^T\left|\sigma_{\Upsilon,s}\right|^2 ds\right)^\frac{p}{2}\right]^\frac{1}{p}<\infty,
			\end{align*}
			where for $z\in \mathbb{R}^{1\times 2}$, $|z|^2:=\operatorname{tr}(zz^T)$.			
			We define $L^p(\tilde \nu)$ as the space of all random fields $U_\Upsilon\colon \Omega\times{[0,T]}\times{[0,l^{L}_{\max}]}\rightarrow \mathbb{R}$ 
			which are measurable with respect to
			$\mathcal{P}\otimes\mathcal{B}([0,l^{L}_{\max}])$ (where $\mathcal{P}$ denotes the predictable $\sigma$-algebra on $\Omega\times[0,T]$ generated
			by the left-continuous $(\mathcal{F}_t)$-adapted processes) such that
			\begin{align*}
				\left\|U_\Upsilon\right\|_{L^p(\tilde \nu) }:=\mathbb{E}\left[\left(\int_0^T\int_{[0,l^{L}_{\max}]}\left|U_{\Upsilon,s}(l)\right|^2 \vartheta(dl)ds\right)^\frac{p}{2}\right]^\frac{1}{p}<\infty.
			\end{align*}
			An $L^p$-solution to a BSDE with terminal condition $\xi$ and generator function $f$ is a triplet $(\Upsilon,\sigma_{\Upsilon},U_{\Upsilon})\in \mathcal{S}^p\times L^p(\bar{W})\times L^p(\tilde N)$ which satisfies for all $t\in{[0,T]}$,
	\begin{equation}\label{eq:bsde_form}
		\begin{aligned}
		\Upsilon_t=\xi&+\int_t^T f(s,\Upsilon_s,\sigma_{\Upsilon,s},U_{\Upsilon,s})ds\\
		&-\int_t^T \sigma_{\Upsilon,s} d\bar{W}_s-\int_{{]t,T]}\times(0,l^{L}_{\max}]}U_{\Upsilon,s}(l)\tilde{\nu}(ds,dl).
		\end{aligned}  
	\end{equation} 

Specifications for the generator and the terminal conditions will be given further below.

		To be able to apply all the necessary BSDE machinery, we need to make for now stronger integrability and boundedness assumptions on the underlying market model. First, we consider a set of assumptions that strengthen the integrability assumption~\eqref{assumption:stochCoeff:integrability}:
			\begin{align}
				&\mathbb{E}\left[\left(\int_0^T{\left(|\lambda_t|+|\sigma_t|^2\right)}dt\right)^2\right]<\infty.\tag{B2}\label{ass:BSDE1}
			\end{align}
			Replacing in the above formulation $2$ by $p>0$:
			\begin{align}
				&\mathbb{E}\left[\left(\int_0^T{\left(|\lambda_t|+|\sigma_t|^{2}\right)}dt\right)^p\right]<\infty.\tag{Bp}\label{ass:BSDE1p}
			\end{align}
			
			%Here we obviously have  $\eqref{assumption:stochCoeff:integrability} \Leftarrow \eqref{ass:BSDE1} \Leftarrow \eqref{ass:BSDE2}$.

\subsection{Analysis of the Generator}\label{rem:BSDEINdiff} 
The following proposition provides the fundamental link between indifference strategies and BSDEs.
\begin{Prop}\label{Prop:LocalIndifferenceStrategiesWithStochCoefficients}
	Assume assumption \eqref{ass:BSDE1} holds and let $\varrho$ be a stopping time with $0\leq \varrho \leq T$, $\hat{\pi}\in {\mathcal{A}}$ a portfolio process and $\hat{\Upsilon} := \Upsilon^{\hat{\pi}}_t$. Then the following are equivalent
	\begin{enumerate}[(i)]
		\item $\hat{\pi}$ is an indifference strategy on $[\varrho,T]\cup\{\infty\}$;
		\item\label{item:indiffBSDEii} There is a process $\sigma_{\Upsilon}\in L^2(\bar{W})$, such that $(\hat{\Upsilon},\sigma_{\Upsilon},0)=(\hat{\Upsilon},\sigma_{\Upsilon},U_{\Upsilon})$ is on $[\varrho,T]$ a solution to the BSDE
		\begin{equation}\label{Prop:LocalIndifferenceStrategiesWithStochCoefficients:BSDE}
		\hat{\Upsilon}_t = \int_t^T\left( \Phi_s(\pi_s^M)- \Phi_s(\hat{\pi}_s)\right)ds -\int_t^T \sigma_{\Upsilon,s}d\bar{W}_s-\int_{(t,T]\times[0,l^{L}_{\max}]}U_{\Upsilon}(s,l)\tilde{\nu}(ds,dl),
		\end{equation}
		where $$\hat{\pi} = \frac{1 - e^{-(\hat{\Upsilon}_t\vee 0)}}{l^{WOC}}\,.$$
		
	\end{enumerate}
\end{Prop} 

Equation \eqref{Prop:LocalIndifferenceStrategiesWithStochCoefficients:BSDE} is a BSDE in the sense of our setting above, \eqref{eq:bsde_form}. To apply existence results from BSDE theory, it is necessary to analyze  its generator, that is the integrand w.r.t.~$ds$. We will do this here by listing several straightforward observations.
					
We expand $\Phi$ in \eqref{item:indiffBSDEii} to relate it to the form \eqref{eq:bsde_form}: The triplet $(\hat{\Upsilon},\sigma_{\Upsilon},0)$ is the solution to \eqref{eq:bsde_form} with terminal condition $\xi=0$ and $f$ given by
\begin{align*}
f(t,y)=&\lambda_t\pi^M_t-\frac{(\sigma_t \pi^M_t)^2}{2}+\int_{[0,l^{L}_{\max}]}\log\Big(1-\pi^M_tl\Big)\vartheta(dl)\\
&-\frac{\lambda_t\big(1 - e^{-(y\vee 0)}\big)}{l^{WOC}}+\frac{\sigma_t^2\big(1 - e^{-(y\vee 0)}\big)^2}{2l^{WOC}}-\int_{[0,l^{L}_{\max}]}\log\Big(1-\frac{1 - e^{-(y\vee 0)}}{l^{WOC}}l\Big)\vartheta(dl).
\end{align*}	

The generator depends on $\pi^M$, that is, on $\lambda$ and $\sigma$ and the terminal condition are measurable with respect to the $\sigma$-algebra generated by $(\tilde{W},\hat{W})$ only. Also the terminal condition $\xi=0$ is trivially measurable by the same $\sigma$-algebra. Therefore the setting might be reduced to BSDEs driven only by a Brownian motion, 
$$\hat{\Upsilon}_t = \int_t^Tf(s,\hat{\Upsilon}_s)ds -\int_t^T \sigma_{\Upsilon,s}d\bar{W}_s,\quad t\in [0,T].$$
However, our setting including the additional $U_\Upsilon$-variable is more natural, as the sources of randomness for the whole model are based on the $\sigma$-algebra $\mathcal{F}$, which includes the jumps. In the section that follows, taking along the jumps is an extension for the BSDE results showing that this part of the theory is feasible for future treatment when considering models that may have jumps in the portfolio processes $\pi$. Note that 
\begin{align*}
f(t,0)=\lambda_t\pi^M_t-\frac{(\sigma_t\pi^M_t)^2}{2}+\int_{[0,l^{L}_{\max}]}\log\Big(1-\pi^M_tl\Big)\vartheta(dl)
\end{align*}
and, due to the boundedness of $\pi^M$, 
\begin{align}\label{eq:fcondI}
\mathbb{E}\left[\left(\int_0^T{|f(t,0)|}dt\right)^p\right]<\infty
\end{align}
holds if \eqref{ass:BSDE1p} is satisfied. In the model without jumps, where $\pi^M=\frac{\lambda}{\sigma^2}\vee 0=\frac{\lambda^+}{\sigma^2}$, 
$$\mathbb{E}\left[\left(\int_0^T{|f(t,0)|}dt\right)^p\right]<\infty$$
holds if \eqref{ass:BSDE1p} and \eqref{assumption:stochCoeff:MertonBounded} do, or, more generally, if \eqref{ass:BSDE1p} is replaced by

\begin{align}
				&\mathbb{E}\left[\left(\int_0^T{\frac{|\lambda^+_t\lambda_t-\frac{1}{2}\lambda^+_t|}{\sigma^2_t}}dt\right)^p\right]<\infty.\tag{BpW}\label{ass:BSDE1pW}
			\end{align}
This condition also follows, if e.g.~$\mathbb{E}\left[\left(\int_0^T\lambda_t^2dt\right)^p\right]<\infty$ and $\tilde\sigma=\inf\{|\sigma_t(\omega)|:(t,\omega)\in [0,T]\times\Omega \}>0$. For the generator $f$ one readily obtains the relation 
\begin{align*}
(f(t,y_1)-f(t,y_2))(y_1-y_2)\leq (\lambda_t^-+\sigma^2)\left(\frac{1}{l^{WOC}}\vee\frac{1}{2(l^{WOC})^2}\right)|y_1-y_2|^2.
\end{align*}
Hence, a one-sided Lipschitz condition 
\begin{align}\label{eq:fcondLip}
(f(t,y_1)-f(t,y_2))(y_1-y_2)\leq K|y_1-y_2|^2.
\end{align}
with $K\in (0,\infty)$ is satisfied for $f$ whenever the processes $\left(\frac{1}{l^{WOC}}\vee\frac{1}{2(l^{WOC})^2}\right)\lambda^-$ and $\left(\frac{1}{l^{WOC}}\vee\frac{1}{2(l^{WOC})^2}\right)\sigma^2$ are bounded by $K$. The conditions \eqref{eq:fcondI} and \eqref{eq:fcondLip} are standard assumptions for BSDEs to yield unique solutions. However, \eqref{eq:fcondLip} demands a strong condition (uniform boundedness) on the $\sigma$ coefficient. We are going to ease that condition in Section \ref{ssec:TreatBSDEs}.

			\begin{proof}[Proof of Proposition \ref{Prop:LocalIndifferenceStrategiesWithStochCoefficients}]	 
				First note that since $\pi_{\infty}=0$ for all admissible strategies, $\hat{\Upsilon}_{\infty} = 0$ necessarily has to hold. Since $\mathcal{F}_T=\mathcal{F}_{\infty}$ and all summands in the definition of $Z:=Z^{\hat{\pi}}$ except for $\hat{\Upsilon}$ are identical for $t=T$ and $t=\infty$, $Z$ is a martingale on the time domain $\{T,\infty\}$, if and only if  the terminal condition $\hat{\Upsilon}_T=0$ holds. We can thus assume $\Upsilon_T=0$ in the remaining part of the proof.
				
				\smallskip
				
				Next, let $\tau$ be an arbitrary stopping time with $\varrho\leq\tau\leq T$. By definition of $Z$ we have
				\begin{equation}\label{Prop:LocalIndifferenceStrategiesWithStochCoefficients:proof:eq1}
					\Upsilon_{\tau} = \int_0^\tau{(\Phi_s(\hat{\pi}_s)-\Phi_s(\pi_s^M))ds} - Z_{\tau} = \Upsilon_{\varrho} + \int_\varrho^\tau{(\Phi_s(\hat{\pi}_s)-\Phi_s(\pi_s^M))ds} - \left(Z_{\tau} - Z_{\varrho}\right).
				\end{equation}
				
				\smallskip
				
				Now if (i) holds, then by the martingale representation theorem
				%\footnote{Note that due to $\Upsilon_T=0$ we have $$Z_T=\int_0^T{\left(\Phi_t(\hat{\pi}_t)-\Phi_t(\pi_t^M)\right)dt}$$ and it is easily seen that assumption \eqref{ass:BSDE1} together with admissibility ensures square-integrability of this random variable. Hence, the martingale representation theorem is applicable here. This is the only instance where we make use of assumption \eqref{ass:BSDE1}, the remainder of the proof just requires \eqref{assumption:stochCoeff:integrability}.}
				there are  processes $\sigma_\Upsilon\in L^2(\bar{W}), U_\Upsilon\in L^2(\tilde{\nu})$, such that for all such stopping times $\tau$\footnote{The martingale representation theorem is often only stated for a deterministic time domain like $[0,T]$, which is sufficient to permit our usage on the interval $[\varrho,T]$: just apply the theorem to the martingale $\hat{Z}_t:=E[Z_T\mid \mathcal{F}_t]$ and use the fact that $\hat{Z}$ and $Z$ must coincide on $[\varrho,T]$.}
				$$Z_{\tau} = Z_{\varrho} - \int_{\varrho}^{\tau}{\sigma_{\Upsilon,t}d\bar{W}_t}-\int_{(\varrho,\tau]\times{(0,l^{L}_{\max}]}}U_{\Upsilon,t}(l)\tilde{\nu}(dt,dl).$$
				Substituting this into (\ref{Prop:LocalIndifferenceStrategiesWithStochCoefficients:proof:eq1}) and using $ \hat{\Upsilon}_t = -\log{(1-l^{WOC}\hat{\pi}_t)} \Leftrightarrow \hat{\pi}_t = \frac{1 - e^{-\hat{\Upsilon}_t}}{l^{WOC}}$ shows that $\hat{\Upsilon}$ satisfies (\ref{Prop:LocalIndifferenceStrategiesWithStochCoefficients:BSDE}). As $\hat{\pi}$ is continuous, also $Z$ is, and so the integrand of the jump part, $U_{\Upsilon}$ must be 0.
				
				\smallskip
				
				Now assume conversely that (ii) holds. Then we obtain from equation (\ref{Prop:LocalIndifferenceStrategiesWithStochCoefficients:proof:eq1})\begin{align*}Z_{\tau}-Z_{\varrho} = -\Upsilon_{\tau} +\Upsilon_{\varrho} + \int_\varrho^\tau{(\Phi_s(\hat{\pi}_s)-\Phi_s(\pi_s^M))ds} = - \int_{\varrho}^{\tau}{\sigma_{\Upsilon,t}d\bar{W}_t}
				\end{align*}
				for any stopping time $\tau\geq \varrho$. 
				Since $\sigma_\Upsilon\in L^2(\bar{W})$, the stochastic integral on the right is a martingale and, hence, $Z$ must be a martingale on $[\varrho,T]$.
				
			\end{proof}
			
			\begin{Rem}
			Without assumption \eqref{ass:BSDE1}, the proof of the direction (i)$\Rightarrow$(ii) in the preceding proposition fails, because then $Z$ is not necessarily a square-integrable martingale and thus the martingale representation theorem does not imply square integrability of the process $\sigma_{\Upsilon}$ anymore. The reader easily verifies that \eqref{ass:BSDE1} was only used in the above proof to conclude $\sigma_{\Upsilon}\in L^2(\bar{W})$. However, if one replaces the part $\sigma_{\Upsilon}\in L^2(\bar{W})$ from (ii) by $\sigma_{\Upsilon}\in L^p(\bar{W})$, for some $p>1$, then (i)$\Rightarrow$(ii) still holds under the assumption \eqref{ass:BSDE1p} for $p>1$ (just set $f=0$ in \cite[Theorem 3.3]{KremsnerSteinicke} or \cite[Theorem 2]{Kruse}).
				%In the model without jumps, if one removes the condition $\sigma_{\Upsilon}\in L^2(\bar{W})$, 
				The assumption  \eqref{assumption:stochCoeff:integrability} alone already ensures that (i)$\Rightarrow$(ii) still holds yielding $\sigma_\Upsilon\in L^\beta(\bar{W})$ for all $\beta\in (0,1)$. (see \cite[Theorem 6.3]{briand2003lp}). %In the L\'evy case we believe the assumptions to be reducible to $p=1$ (like in the case without jumps) as well, inspecting the proofs of Theorems 6.2 and 6.3 from \cite{briand2003lp} and just adapting them to the L\'evy setting by the use of \cite[Proposition 4.4]{KremsnerSteinicke} and the generalized It\^o formula \cite[Lemma 7]{Kruse}.
			\end{Rem}
			Consequently, the BSDE representation \eqref{Prop:LocalIndifferenceStrategiesWithStochCoefficients:BSDE} with a generic process $\sigma_{\Upsilon}$ still holds for any indifference strategy $\hat{\pi}$, even without assumption \eqref{ass:BSDE1}. 
%			This immediately implies the following:
%			
%		
%				\begin{Prop}
%					Suppose that in the setting without jumps assumption \eqref{assumption:stochCoeff:integrability} holds, and in the L\'evy-setting, \eqref{ass:BSDE1p} for some $p>1$ is satisfied. Then any indifference strategy $\hat{\pi}\in \tilde{\mathcal{A}}$ is right-continuous. In particular, any indifference strategy is admissible (i.e. in $\mathcal{A}$). {\color{red} Alex: Check if we need right continuity at all in what follows and in the chapter before.}
%				\end{Prop}

				In the next subsections, we show that BSDE (\ref{Prop:LocalIndifferenceStrategiesWithStochCoefficients:BSDE}) has a unique square-integrable solution, assuming that $\lambda$ and $\sigma$ satisfy the following condition.
				\begin{align}\text{For some }\varepsilon>0,\quad
					\mathbb{E}\left[\int_0^T\exp\left(\varepsilon\left(\lambda_t^-+\sigma_t^2\right)\right)dt\right]<\infty\,.\tag{Bexp}
					\label{ass:BSDEXPB}
				\end{align}
				Note that if $\lambda^-$ is replaced by $\lambda$, the  above assumption implies \eqref{ass:BSDE1p} for all $p>0$ and therefore also \eqref{ass:BSDE1} and \eqref{assumption:stochCoeff:integrability}.
			
			\bigskip

					\subsection{Solutions to the Indifference Utility BSDE}\label{ssec:TreatBSDEs}
					We present various existence and uniqueness results about BSDEs here which are relevant for the above characterization. We will also show a comparison theorem for use in Section \ref{sec:Markovian}.
The theorems in this subsection consider the full L\'evy setting. All assertions hold for the case without jumps as well thanks to \eqref{assumption:stochCoeff:MertonBounded} which grants boundedness for the Merton strategy. Alternatively, the theorems also work assuming \eqref{ass:BSDE1pW}, and, in place of \eqref{ass:BSDEXPB}, using $\mathbb{E}\left[\int_0^T\exp\left(\varepsilon\frac{|\lambda^+_t\lambda_t-\frac{1}{2}\lambda^+_t|}{\sigma^2_t}\right)dt\right]<\infty$ . For a better readability, the corresponding proofs are delegated to Appendix~\ref{App:Proofs}.
					\begin{Prop}\label{Prop:ExistenceOfIndifferenceBSDESolutionJumps2}
						If \eqref{ass:BSDE1} holds and $\lambda^+,\sigma,$ are bounded processes, then there is a unique triplet $(\Upsilon,\sigma_{\Upsilon},U_\Upsilon)\in \mathcal{S}^2\times L^2(\bar{W})\times L^2(\tilde{\nu})$
						which solves the BSDE (\ref{Prop:LocalIndifferenceStrategiesWithStochCoefficients:BSDE}). 
					\end{Prop}
					We can now obtain solutions of unbounded $\lambda,\sigma$. By an approximation procedure for the BSDE's generator, along with standard methods for BSDEs with $L^2$-data, we get the following:  
					\begin{Thm}\label{Thm:ExistenceOfIndifferenceBSDESolutionUnbounded2Jumps}
						If \eqref{ass:BSDE1} holds
						then there is a triplet $(\Upsilon,\sigma_{\Upsilon},U_\Upsilon)\in\mathcal{S}^2\times L^2(\bar{W})\times L^2(\tilde{\nu})$ which solves the BSDE (\ref{Prop:LocalIndifferenceStrategiesWithStochCoefficients:BSDE}). 
					\end{Thm}

					\begin{Rem}\label{rem:BSDErem}
						In a similar manner, one can show that if for some $p>1$, \eqref{ass:BSDE1p} is satisfied then there is a triplet $(\Upsilon,\sigma_\Upsilon,U_\Upsilon)\in \mathcal{S}^p\times L^p(\bar{W})\times L^p(\tilde{\nu})$ solving the BSDE \eqref{Prop:LocalIndifferenceStrategiesWithStochCoefficients:BSDE}.
						In the case without jumps, by using \cite[Theorems 6.2 and 6.3]{briand2003lp} we can show that assuming only \eqref{assumption:stochCoeff:integrability}, the above Theorem holds with $(\Upsilon,\sigma_\Upsilon)\in \mathcal{D}\times \bigcup_{\beta\in (0,1)}L^\beta(\bar{W})$.

						As mentioned below in the proof of Proposition \ref{Prop:LocalIndifferenceStrategiesWithStochCoefficients}, it is possible to derive an equivalent of \cite[Theorems 6.2 and 6.3]{briand2003lp} for the L\'evy case using methods from there, \cite[Proposition 4.4]{KremsnerSteinicke} and the It\^o formula of \cite{Kruse} to find that already under \eqref{assumption:stochCoeff:integrability}, we get a solution $(\Upsilon,\sigma_\Upsilon,U_\Upsilon)\in\mathcal{D}\times\bigcup_{\beta\in (0,1)}L^\beta(\bar{W})\times \bigcup_{\beta\in (0,1)}L^\beta(\tilde{\nu})$.
					\end{Rem}	
	With stronger tail properties, in addition to existence, the following theorem grants uniqueness of solutions in a class of functions.
					\begin{Thm}\label{Thm:UniquenessOfIndifferenceBSDESolutionUnboundedJumps}
						If \eqref{ass:BSDEXPB} holds and if $\mathbb{E}\left[\left(\int_0^T|\lambda^+_t|dt\right)^p\right]<\infty$ for some $p\in (0,\infty)$, then there is a unique triplet of progressively measurable functions $(\Upsilon,\sigma_{\Upsilon},U_\Upsilon) \in \mathcal{S}^p\times L^p(\bar{W})\times L^p(\tilde{\nu})$ , which solves the BSDE (\ref{Prop:LocalIndifferenceStrategiesWithStochCoefficients:BSDE}). 
					\end{Thm}

					The next theorem states that for generators similar to the function $f$, we may compare the solutions to BSDEs if we have an inequality for the data.

					\begin{Thm}\label{Thm:ComparisonOfIndifferenceBSDESolutionUnboundedJumps}
						
						Let $f,f'\colon \Omega\times[0,T]\times \mathbb{R}\to\mathbb{R}$ be two (measurable) generators, such that there is a progressively measurable, non-negative process $K$ and $\varepsilon\in(0,\infty)$ with
						$$\mathbb{E}\left[\int_0^T\exp\left(\varepsilon K_t\right)dt\right]<\infty,$$
						and that for all $y,y'\in \mathbb{R}$
						$$|f(t,y)-f(t,y')|\leq K_t|y-y'|\wedge K_t, \quad(t,\omega)\text{-a.e.}$$ 
						Assume that $\xi,\xi'\in L^2$ and that there are solution triplets $(\Upsilon,\sigma_{\Upsilon},U_\Upsilon), (\Upsilon',\sigma_{\Upsilon}',U_\Upsilon')\in \mathcal{S}^2\times L^2(\bar{W})\times L^2(\tilde{\nu})$ to the BSDEs given by the terminal conditions $\xi,\xi'$ and generators $f,f'$.
						We assert that, if $\xi\leq \xi'$ a.s. and $f(t,\Upsilon'_t)\leq f'(t,\Upsilon'_t)$, $(\omega,t)$-a.e., then also $\Upsilon_t\leq \Upsilon'_t$ a.s. for all $t\in[0,T]$.\medskip

					\end{Thm}
					\begin{Rem}
						Alternatively, to the assumptions on $f$, one may assume instead that that there is a progressively measurable, non-negative process $K'$ and $\varepsilon\in(0,\infty)$ with
						$$\mathbb{E}\left[\int_0^T\exp\left(\varepsilon K'_t\right)dt\right]<\infty,$$
						and that for all $y,y'\in \mathbb{R}$
						$$|f'(t,y)-f'(t,y')|\leq K'_t|y-y'|\wedge K'_t, \quad(s,\omega)\text{-a.e.}$$ and that $f(t,\Upsilon_t)\leq f'(t,\Upsilon_t)$. The same assertion as above holds also in this case.
					\end{Rem}

\subsection{Existence of Indifference Strategies}				
		
We now have everything at hands in order to prove the existence of a unique indifference strategy. In particular, combining Proposition \ref{Prop:LocalIndifferenceStrategiesWithStochCoefficients} with $\varrho=0$ and Theorem \ref{Thm:UniquenessOfIndifferenceBSDESolutionUnboundedJumps} implies the following uniqueness result. %\footnote{Theorem \ref{Thm:UniquenessOfIndifferenceBSDESolutionUnboundedJumps} asserts uniqueness of the exposure process $\hat{\Upsilon}$, but due to the one-to-one mapping between a portfolio process $\pi$ and its exposure process $\Upsilon^\pi$ this also implies uniqueness of the associated indifference strategy $\hat{\pi}$. As is already known from the constant market coefficient case (compare e.g.~\citet{Seifried2010}, Remark 4.2) this is a special feature of the one-dimensional worst-case problem.}
					
					\begin{Cor}[Uniqueness of indifference strategies]\label{Cor:ExistenceAndUniquenessOfIndifferenceStrategies}
						Under the assumption \eqref{ass:BSDEXPB} there is a unique indifference strategy $\hat{\pi}$.
					\end{Cor}
				
				The indifference strategy $\hat{\pi}$ is also a special superindifference strategy and we have seen above that superindifference strategies are only helpful to bound the worst-case optimal solution, if they are uniformly dominated by the post-crash optimal strategy. It is thus natural to ask, whether this is the case for $\hat{\pi}$. While we do not provide a formal counterexample, our numerical experiments suggest, that one cannot hope for this to be the case in general. However, in certain situations this uniform dominance condition is indeed obtained. We provide here two sufficient conditions for this to happen. 				
				The first is almost trivial, but only valid in market models without jumps with very large excess returns relative to (Brownian) risk.
				
				\begin{Lem}\label{Lem:largeMertonStrategy}
					If $\lambda\geq \frac{\sigma^2}{l^{WOC}}$, then $\hat{\pi}\leq \pi^M$.
				\end{Lem}
				\begin{proof}
					Due to $\pi^M = \frac{\lambda}{\sigma^2}\vee0$, we immediately obtain $\pi^M \geq \frac{1}{l^{WOC}}$. On the other hand, any pre-crash-admissible strategy is bounded above by $\frac{1}{l^{WOC}}$, in particular also the indifference strategy $\hat{\pi}$.
				\end{proof}
				
				%If the Merton strategy is not large enough to be nowhere pre-crash admissible, then 
				In the model with jumps, the following proposition establishes the desired dominance condition. In particular, if $\pi^M$ is obtained by first finding a maximizer $\operatorname{argmax}_{\mathbb{R}}\Phi$, and then trimming it to the range $\big[0,\frac{1}{l^L_{\max}}\big]$.

					\begin{Prop}\label{Prop:MertonSubindifference}
						If the post-crash optimal portfolio process $\pi^M$ is obtained via an It\^o process $\varpi^M$, with coefficients $a,b$,
						\begin{align*}
						&d\varpi^M_t=a_tdt+b_td\bar{W}_t,\quad \varpi_0\in\big[0,\frac{1}{l^L_{\max}}\big] \\
						&\pi^M=(\varpi^M\vee 0)\wedge\frac{1}{l^L_{\max}}
						\end{align*}
						where $\mathbb{E}\left[\int_0^T a_s^2ds\right]<\infty$ and the stochastic integral $\int_0^{\cdot}b_td\bar{W}_t$ is a martingale and if $\pi^M$  
						\begin{enumerate}
				\item is also pre-crash-admissible (i.e. the minimum in the above equation has no effect),
				\item satisfies
				\begin{align*}
				\mathbb{E}\biggl[\int_0^T \biggl(\left|\frac{l^{WOC}{\bf 1}_{[0,\infty)}(\varpi_s)a_s}{1-l^{WOC}\pi^M_s}\right|+&\left|\frac{\left(l^{WOC}\right)^2{\bf 1}_{[0,\infty)}(\varpi_s)|b_s|^2}{(1-l^{WOC}\pi^M_s)^2}\right|\biggr)\biggr]<\infty,
				\end{align*}
				\item and is a subindifference strategy,
\end{enumerate}						  
then $\hat{\pi}_t\leq \pi^M_t$ $\mathbb{P}$-a.s.\ for all $t\in [0,T]$.
				\end{Prop}

				The sufficient conditions of Lemma \ref{Lem:largeMertonStrategy} and Proposition \ref{Prop:MertonSubindifference} are rather restrictive or difficult to apply in many situations. In particular, inspecting the 'almost It\^o process' from the proof of Proposition  \ref{Prop:MertonSubindifference} in Appendix \ref{proof:propaAlmostIto}, to check whether $\pi^M$ being a subindifference strategy requires to find out if
				
$$-\mathbb{E}\left[\int_{t_1}^{t_2}a_{\Upsilon,s}ds\middle|\mathcal{F}_{t_1}\right]+\frac{(l^{WOC})^2}{2}\left(\mathbb{E}\left[L^\varpi_{t_2}\middle|\mathcal{F}_{t_1}\right]-L^\varpi_{t_1}\right)\leq 0,$$
where $a_{\Upsilon,s}$ is again the integrand w.r.t.~$ds$ in \eqref{eq:piMsubdiff}. 			

On the one hand, for any particular model one thus needs to come up with additional arguments why this proposition is applicable. On the other hand, however, these results are useful in the important special case of a constant post-crash optimal strategy, that is, if assumption \eqref{assumption:stochCoeff:MertonConstant} is satisfied. In this case we can conclude:
				
				\begin{Cor}
					Under assumptions \eqref{ass:BSDE1} and \eqref{assumption:stochCoeff:MertonConstant}, $\hat{\pi}_t\leq \pi_t^M$ holds $\mathbb{P}$-a.s.~for all $t\in[0,T]$, and therefore $\hat{\pi}$ is a pre-crash optimal strategy.
				\end{Cor}
			 These existence results will also be of use in the following sections when dealing with concrete examples and numerical investigations.			
\section{The Markovian Case}\label{sec:Markovian}

			In this section we specialise on a market model with $\sigma_{t}=\sigma(z_{t})$, $\lambda_{t}=\lambda(z_{t})$ where $z$ is a factor process whose evolution is governed by the SDE 
			\begin{equation}
				dz_{t}=\mu(z_{t})dt+\varsigma(z_{t})d\hat{W}_{t}\label{eq:stateSDE}
			\end{equation}
			and $z_{0}$ is a fixed initial value.
			We require that the functions $\mu:\mathbb{R}\rightarrow\mathbb{R}$ and $\varsigma:\mathbb{R}\rightarrow D \subseteq \mathbb{R}$ are chosen in a way to guarantee that this SDE has a unique (strong) solution. Furthermore, the functions $\sigma$, $\lambda$, $\mu$ and $\varsigma$ are all assumed to be continuous. We wish to stress that this formulation covers in particular the Bates model, respectively the Heston, c.f.~\eqref{eq:Bates}.
				In this situation, the indifference BSDE \eqref{Prop:LocalIndifferenceStrategiesWithStochCoefficients:BSDE} can be connected to a partial differential equation (PDE). 
				\subsection{From BSDEs to associated PDEs}
				Consider an arbitrary process $\Upsilon_{t}:=v(t,z_{t})$ with some function $v\in C^{1,2}$. Then by It\^o's formula
				\begin{align*}
					d\Upsilon_{t}=&\Bigg(\partial_{t}v(t,z_{t})+\mu(z_{t})\partial_{x}v(t,z_{t})+\frac{\varsigma(z_{t})^{2}}{2}\partial_{xx}v(t,z_{t})\Bigg)dt+\varsigma(z_{t})\partial_{x}v(t,z_{t})d\hat{W}_{t}.
				\end{align*}
				Let the optimal post-crash strategy $\pi^M$ be expressed by a function $\psi$ dependent on $(\lambda,\sigma)$, $\pi^M=\psi(\lambda,\sigma)$, e.g.~in the model without jumps, $\psi(\lambda,\sigma)=\frac{\lambda}{\sigma^2}\vee 0.$ In models involving jumps, expressions for $\psi(\lambda,\sigma)$ may also be obtained explicitly, in the case of the L\'evy measure $\vartheta$ equals the point measure $\delta_{l^L_{\max}}$, solving $\partial_y\Phi(y)=\lambda-\sigma^2 y-\frac{l^L_{\max}}{1-y l^L_{\max}}$, one finds the maximizer $\psi(\lambda,\sigma)=\frac{\lambda l^L_{\max}+\sigma^2-\sqrt{(l^L_{\max})^2(\lambda^2+4\sigma^2)-2l^L_{\max}\lambda\sigma^2+\sigma^4}}{2l^L_{\max}\sigma^2}$, see also Section \ref{subsec:jumpModel} for a similar computation.

				With $\pi^M$ being given by $\psi(\lambda,\sigma)$, the drift of the forward SDE equals the generator of the BSDE \eqref{Prop:LocalIndifferenceStrategiesWithStochCoefficients:BSDE}, $f(t,v(t,z_t))$, if
				\begin{align*}
					&\partial_{t}v(t,x)+\mu(x)\partial_{x}v(t,x)+\frac{\varsigma(x)^{2}}{2}\partial_{xx}v(t,x)\\
					&+\lambda(x)\psi(\lambda(x),\sigma(x))-\frac{\sigma(x)^2}{2}\psi(\lambda(x),\sigma(x))^2+\int_{[0,l^{L}_{\max}]}\log\left(1-\psi(\lambda(x),\sigma(x))l\right)\vartheta(dl)\\
					&-\lambda(x)\frac{1-e^{-(v(t,x)\vee 0)}}{l^{WOC}}+\frac{\sigma(x)^{2}}{2}\left(\frac{1-e^{-(v(t,x)\vee 0)}}{l^{WOC}}\right)^{2}\\
					&-\int_{[0,l^{L}_{\max}]}\log\left(1-\frac{1-e^{-(v(t,x))\vee 0)}}{l^{WOC}}l\right)\vartheta(dl)\\
					&=0
				\end{align*}
				holds for all $(t,x)\in[0,T]\times\mathbb{R}$ since
				\begin{align*}
					&\Phi_{t}({\pi_t^M})-\Phi_{t}\left(\frac{1-e^{-(\Upsilon_{t}\vee 0)}}{l^{WOC}}\right)\\
					&\ =(\Phi_t(\psi(\lambda(z_t),\sigma(z_t))-r_t)-\lambda(z_{t})\frac{1-e^{-(\Upsilon_{t}\vee 0)}}{l^{WOC}}+\frac{\sigma(z_{t})^{2}}{2}\left(\frac{1-e^{-(\Upsilon_{t}\vee 0)}}{l^{WOC}}\right)^{2}\\
					& \quad-\int_{[0,l^{L}_{\max}]}\log\left(1-\frac{1-e^{-(\Upsilon_t\vee 0)}}{l^{WOC}}l\right)\vartheta(dl),
				\end{align*}
				recalling that
				\begin{align*}
				\Phi_t(y)=r_t+\lambda(z_t)y-\frac{\sigma(z_t)^2}{2}y^2+\int_{[0,l^{L}_{\max}]}\log\left(1-yl\right)\vartheta(dl).
\end{align*}				 
				In this case $\Upsilon$ satisfies the indifference BSDE \eqref{Prop:LocalIndifferenceStrategiesWithStochCoefficients:BSDE}, if in addition $\Upsilon_{T}=v(T,z_{T})=0$. A sufficient condition for this is $v(T,x)=0$ for all $x\in\mathbb{R}$. Thus, the following result holds (cf.~\cite[Proposition 4.3]{ElKarouiPengQuenez1997}):
			
			\begin{Prop}\label{Prop:indifferenceExposurePDE} 
			Let $v\in C^{1,2}$ be a solution to the PDE
			\begin{equation}\label{eq:indifferencePDE}
				\begin{aligned}
					&\partial_{t}v(t,x)+\mu(x)\partial_{x}v(t,x)+\frac{\varsigma(x)^{2}}{2}\partial_{xx}v(t,x)\\
					&+(\Phi_t(\psi(\lambda(x),\sigma(x))-r_t)-\lambda(x)\frac{1-e^{-(v(t,x)\vee 0)}}{l^{WOC}}+\frac{\sigma(x)^{2}}{2}\left(\frac{1-e^{-(v(t,x)\vee 0)}}{l^{WOC}}\right)^{2}\\
					&\quad-\int_{[0,l^{L}_{\max}]}\log\left(1-\frac{1-e^{-(v(t,x)\vee 0)}}{l^{WOC}}l\right)\vartheta(dl)=0,\\
					&\quad \quad v(T,x)=0\,.
				\end{aligned}
			\end{equation}
					and suppose that $\Upsilon_{t}:=v(t,z_{t})$, $\sigma_{\Upsilon,t}:=\begin{pmatrix}\varsigma(z_{t})\partial_{x}v(t,z_{t}) & 0\end{pmatrix}$ and $U_{\Upsilon,t}=0$ are processes in $\mathcal{S}^1\times L^1(\bar{W})\times L^1(\tilde{\nu})$. Then $(\Upsilon,\sigma_{\Upsilon},U_\Upsilon)$ is a solution to the BSDE (\ref{Prop:LocalIndifferenceStrategiesWithStochCoefficients:BSDE}), with the generator now having the form
					\begin{align*}
						f(t,y)=&(\Phi_t(\psi(\lambda(z_t),\sigma(z_t))-r_t)-\lambda(z_t)\frac{1-e^{-(y\vee 0)}}{l^{WOC}}+\frac{\sigma(z_{t})^{2}}{2}\left(\frac{1-e^{-(y\vee 0)}}{l^{WOC}}\right)^{2}\\
						&-\int_{[0,l^{L}_{\max}]}\log\left(1-\frac{1-e^{-(y\vee 0))}}{l^{WOC}}l\right)\vartheta(dl).
				\end{align*}
			\end{Prop} 
			Note that the second component of $\sigma_\Upsilon$ is zero as well as $U_\Upsilon$. This is due to the fact that all relevant processes $\sigma(z), \lambda(z), \pi^M, z$ are measurable with respect to the filtration generated by $\hat{W}$ alone. So the generator $f$ is also measurable w.r.t.~the filtration generated by $\hat{W}$, as is $\Upsilon_T=0$ and thus for solvability of the BSDE no integrands w.r.t.~$\tilde{W}$ and $\tilde{\nu}$ are needed.
			
			The last proposition associates the indifference BSDE and thus ultimately the model's unique indifference strategy and thus its worst-case optimal solution with a PDE. 
			
			While instructive, Proposition \ref{Prop:indifferenceExposurePDE} is usually of little practical value, because PDE \eqref{eq:indifferencePDE} can in most cases only be solved numerically and then one needs an a-priori argument why a classical solution to the PDE exists and why the numerical scheme employed converges to such a classical solution. In the following we therefore consider cases, in which a version of Proposition \ref{Prop:indifferenceExposurePDE} still holds true, if $v$ is only a viscosity solution to \eqref{eq:indifferencePDE}.
			
			The convergence of the numerical scheme itself holds due to \cite{BarlesSouganidis1991}.
			To adapt our notation to Markovian BSDEs, we highlight the dependence on the forward process $z$ in $f$ by a variable $x$. So, we consider another generator $\tilde{f}$ in the variables $(t,x,y)$ such that $f(\omega, t,y)=\tilde{f}(t,z_t(\omega),y)$. Abusing notation slightly, we identify $\tilde{f}=f$. In our setting, it is given by the function
			\begin{align*}
				&f(t,x,y)=\\
				&\lambda(x)\psi(\lambda(x),\sigma(x))-\frac{\sigma(x)^2}{2}\psi(\lambda(x),\sigma(x))^2+\int_{[0,l^{L}_{\max}]}\log\left(1-\psi(\lambda(x),\sigma(x))l\right)\vartheta(dl)\\
				&-\lambda(x)\frac{1-e^{-(y\vee 0)}}{l^{WOC}}+\frac{1}{2}\left(\sigma(x)\frac{1-e^{-(y\vee 0)}}{l^{WOC}}\right)^{2}-\int_{[0,l^{L}_{\max}]}\log\left(1-\frac{1-e^{-(y\vee 0)}}{l^{WOC}}l\right)\vartheta(dl),
			\end{align*}
			and in the case without jumps,
			$$f(t,x,y)=\frac{\lambda(x)\lambda(x)^+}{\sigma(x)^2}-\frac{(\lambda(x)^+)^2}{2\sigma(x)^2}-\lambda(x)\frac{1-e^{-(y\vee 0)}}{l^{WOC}}+\frac{1}{2}\left(\sigma(x)\frac{1-e^{-(y\vee 0)}}{l^{WOC}}\right)^{2}.$$
			 Note that in contrast to \cite{ElKarouiPengQuenez1997,bbp}, the driver $f$ of our BSDE is not Lipschitz in $y$, but satisfies the prerequisites of Theorems \ref{Thm:UniquenessOfIndifferenceBSDESolutionUnboundedJumps} and  \ref{Thm:ComparisonOfIndifferenceBSDESolutionUnboundedJumps}.
			The following theorem ensures existence and uniqueness of solutions to the corresponding PDE in our setting (the proof can be found in Appendix~\ref{App:Markov}):
			
			%As in \cite{ElKarouiPengQuenez1997}, we assume on the coefficients
			%of the state process SDE that $\mu$ and $\varsigma$ are globally Lipschitz, cf.~also \cite{bbp}. We will generalize this to a certain extent.
			
				\begin{Thm}\label{Thm:existenceOfViscositySolutionJumps}
					Let $z$ be the solution of the forward SDE \eqref{eq:stateSDE} on $[t,T]$, starting in $x$ at $t$ that satisfies the following conditions: For all $p\geq 2$ there exists a constant $M_p$ such that 
					\begin{equation}\label{ass:1}
						\mathbb{E}\left[\sup_{t\leq r\leq s}|z_r^t(x)-x|^p\right]\leq M_p(s-t)(1+|x|^p)
					\end{equation}
					and
					\begin{equation}\label{ass:2}
						\mathbb{E}\left[\sup_{t\leq r\leq s}|z_r^t(x)-z_r^t(x')-(x-x')|^p\right]\leq M_p(s-t)(|x-x'|^p+|\sqrt{x}-\sqrt{x'}|^p)
					\end{equation}
					We furthermore assume that for some $C,p\in (0,\infty)$,
					$$|\lambda(x)|+|\sigma(x)^2|\leq C(1+|x|^p),$$
					and there is an $\varepsilon>0$ such that
					\begin{equation}\label{ass:3}
						\mathbb{E}\left[\int_0^T\exp\left(\varepsilon\left(\lambda(z_t)^-+\sigma(z_t)^2\right)\right)dt\right]<\infty
					\end{equation}
					(i.e.~ \eqref{ass:BSDEXPB} holds for $\lambda\equiv\lambda(z)$, $\sigma^2\equiv\sigma^2(z)$).

					Then, there exists a viscosity solution $v$ of \eqref{eq:indifferencePDE}.
					If moreover $\lambda$ and $\sigma$ satisfy for all $R>0$
					$$|\lambda(x)-\lambda(x')|+|\sigma(x)^2-\sigma(x')^2|\leq m_R(|x-x'|),$$
					for all $|x|,|x'|\leq R $ and for a continuous function $m_R\colon [0,\infty)\to[0,\infty)$ with $m_R(0)=0$, the solution is unique in the class of functions $u$ such that $$\lim_{|x|\to\infty}|u(t,x)|e^{-\tilde{A}(\log(|x|))^2}=0,$$
					uniformly in $t$ for some $\tilde{A}>0$. 
				\end{Thm}

\subsection{Concrete Examples}
As a first example, one may consider $\mu$ and $\varsigma$ to be globally Lipschitz continuous. In this case, the SDE \eqref{eq:stateSDE} has a unique, strong solution. Furthermore, we have that the Assumptions \eqref{ass:1} and \eqref{ass:2} are satisfied (see e.g.~\cite[Proposition 1.1]{bbp},\cite{fujiwara1985}).
Such a choice for $\mu$ and $\varsigma$ is made e.g.~in the Kim-Omberg model, cf.~\cite{KimOmberg1996},

\begin{align*}
dz_t = \kappa(\theta-z_t)dt + \tilde{\varsigma} d\hat{W}_t,
\end{align*}

with constants $\theta, \kappa, \tilde{\varsigma}$. Furthermore, we have $\lambda_t=z_t$ and we have the constant volatility $\sigma_t=\tilde{\sigma}\in (0,\infty)$. Assumption \eqref{ass:3} is satisfied since $\exp(z_t)$ is log-normally distributed.
			
The second example is the Heston/Bates model, cf.~\cite{Heston1993,Bates1996}, where the volatility is given by the CIR process, i.e. 
			\begin{equation}\label{CIR}
				dz_t = \kappa(\theta-z_t)dt + \tilde{\varsigma}\sqrt{z_t} d\hat{W}_t
			\end{equation}
with positive constants $\theta, \kappa, \tilde{\varsigma}$. Here, we have $\sigma_t = \sqrt{z_t}$.
Note that the diffusion coefficient of \eqref{CIR} is not globally Lipschitz continuous and therefore standard results do not apply. Since the square root is of linear growth, the CIR process has a unique strong solution by the Yamada-Watanabe condition (see e.g.~\cite[Theorem IV 3.2]{ikedawatanabe}). Moreover, the process is strictly positive if $\frac{2\kappa\theta}{\tilde{\varsigma^2}}>1$ by the Feller condition. Furthermore, the following Proposition ensures that Assumptions \eqref{ass:1} and \eqref{ass:2} of Theorem~\ref{Thm:existenceOfViscositySolutionJumps} are fulfilled (its proof can be found in Appendix~\ref{App:Markov}).

				\begin{Prop}\label{Lem: CIR-cont}
					Let $\mu(x)=\kappa(\theta-x)$ and $\varsigma(x)=\tilde{\varsigma}\sqrt{x}$ for some $\kappa,\theta,\tilde{\varsigma}>0$. Furthermore, let $\frac{2\kappa\theta}{\tilde{\varsigma}^2}>\frac{1}{2}$. For all $p\geq 2$ there is a constant $M_p$ such that
					\begin{equation}\label{eq:Cir1}
						\mathbb{E}\left[\sup_{t\leq r\leq s}|z_r^t(x)-x|^p\right]\leq M_p(s-t)(1+|x|^p)
					\end{equation}
					and
					\begin{equation}\label{eq:Cir2}
						\mathbb{E}\left[\sup_{t\leq r\leq s}|z_r^t(x)-z_r^t(x')-(x-x')|^p\right]\leq M_p(s-t)(|x-x'|^p+|\sqrt{x}-\sqrt{x'}|^p)
					\end{equation}
					Moreover, in the case of $p=1$ and all $\kappa,\theta,\tilde{\varsigma}>0$, we have that
					\begin{equation}\label{L1_cir}
						\mathbb{E}\left[|z_t(x)-z_t(x')|\right]= e^{-\kappa t}|x-x'|
					\end{equation}
				\end{Prop}
			\begin{Rem}
				Note that Proposition \ref{Lem: CIR-cont} extends the results of Section 4.10.1 of \cite{CHJ}, where a local Lipschitz continuity in the initial value was shown under the condition $\frac{2\kappa\theta}{\tilde{\varsigma}^2}>1$. Similar results to Proposition \ref{Lem: CIR-cont} about the regularity of the CIR process in the initial value can be found also in \cite[Section 4]{Hefter2018}.
			\end{Rem}
Next, we need to check whether \eqref{ass:3} holds. For the CIR process, we have the following result (see e.g.~\cite{AltNeu} and \cite{HK}).
\begin{Prop}
	Let $\frac{2\kappa\theta}{\tilde{\varsigma}^2}\ge 1$. Then,
	\begin{align*}
		\mathbb{E}\left[\exp\left(\varepsilon z_t\right)\right]<\infty \quad \text{iff } \varepsilon<\frac{2\kappa}{\tilde{\varsigma}^2(1-\exp(-\kappa t))}.
	\end{align*}
\end{Prop}
Without the Feller condition, a small calculation using \cite[Proposition 3.2]{CozReis2016} (carried out in Appendix \ref{App:Markov}) shows the more general
\begin{Prop}\label{prop:CIR-exp-integr}
	Let $\varepsilon<\frac{2\kappa}{\tilde{\varsigma}^2}$. Then, $\sup_{t\in [0,T]}\mathbb{E}\left[\exp\left(\varepsilon z_t\right)\right]<\infty.$
\end{Prop}
Since we can choose $\varepsilon>0$ arbitrarily small and $p=1$ in the case of the Heston Model, we conclude that \eqref{ass:3} is fulfilled and we can apply Theorem \ref{Thm:existenceOfViscositySolutionJumps} for the parameter range $\frac{2\kappa\theta}{\tilde{\varsigma}^2}\ge \frac{1}{2}$.

\subsection{Modeling the Jumps}\label{subsec:jumpModel}
To model jump intensities we choose for example the measure $\vartheta_1(dl)=\frac{1}{l}dl\big|_{[0,l^{L}_{\max}]}$ (infinite activity) and $\vartheta_2(dl)=\delta_q$, $q\in [0,l^{L}_{\max}]$ (possibility of smaller crashes of constant size).  To obtain constant optimal post-crash strategies in the two cases, $\pi^{M,1}=\alpha_1, \pi^{M,2}=\alpha_2$, both in $\big(0,\frac{1}{l^{L}_{\max}}\big)$, we want to find for a given Heston volatility $t\mapsto\tilde{\varsigma}\sqrt{z_t}$, the $\psi^1$-(resp. $\psi^2$-)appropriate market prices of risk $\lambda^1$ and $\lambda^2$ in the sense of Definition~\ref{def:appropriate}. To that end, by Proposition \ref{Prob:WC:PreCrashWithStochCoefficients} (and the argumentation for differentiability in the proof), we may obtain  $\alpha^i=\operatorname{argmax}_{[0,l^{L}_{\max}]}\Phi^i$ for $i=1,2$, by solving $\partial_y\Phi^i(\alpha_i)=0$. This means that

$$0=\lambda^i-(\sigma^{i})^2\alpha_i-\int_{[0,l^{L}_{\max}]}\frac{l}{1-\alpha_i l}\vartheta_i(dl).$$ 
Hence, we get
$$\lambda^i=(\sigma^{i})^2\alpha_i+\int_{[0,l^{L}_{\max}]}\frac{l}{1-\alpha_i l}\vartheta_i(dl),$$
which corresponds to a linear price in the volatility plus an additional safety loading for the jump term.

Computing the integrals for the concrete $\vartheta_i$, we see that our $\lambda^i$ need to be
\begin{align*}
\lambda^1(z)&=\sigma^2(z)\alpha_1-\frac{\log(1-\alpha_1 l^{L}_{\max})}{\alpha_1}=z\alpha_1-\frac{\log(1-\alpha_1 l^{L}_{\max})}{\alpha_1}\\
\lambda^2(z)&=\sigma^2(z)\alpha_2+\frac{q}{1-\alpha_2q}=z\alpha_2+\frac{q}{1-\alpha_2q}.
\end{align*}

		\section{Numerical Experiments}\label{sec:numerics}
		
\subsection{Numerics for the Bates and Heston model}	
The first examples we will show here, calculated using methods from the previous section, are variations of the Bates and Heston model with different activities of jumps. All models rely on the same samples of a CIR process, computed in 1000 time steps using distributional properties (exact simulation).
\begin{enumerate}[(a)]
\item\label{it:numa} Infinite activity jumps: $\vartheta_1=\frac{1}{l}dl$. Here our coefficients summarize as follows: The forward equation for the CIR-process modeling the volatility is
\begin{align*}
dz_t=\kappa(\theta-z_t)dt+\tilde{\varsigma}\sqrt{z_t}d\hat{W}_t,\quad z_0=\theta,\quad t\in [0,T],
\end{align*}
with values $\kappa=3.99, \theta=0.014, \tilde\varsigma=0.27, T=5$, corresponding to a Feller index of $1.5325$, so our requirements from Section \ref{sec:Markovian} are fulfilled. The parameters for the CIR process are taken from \cite[Table 6]{BroadieKaya}. The further coefficients for the model are

\begin{align*}
&\pi^{M,1}=\alpha_1=2.5,\quad l^{WOC}=0.5,\quad l^L_{\max}=0.2,\\
&\sigma^1(z)=\sqrt{z},\\
&\lambda^1(z)=z\alpha_1-\frac{\log(1-\alpha_1 l^{L}_{\max})}{\alpha_1}\\
&
\end{align*}
Here the excess returns $\lambda^1$ is the $\psi^1$-appropriate market price of risk from Subsection \ref{subsec:jumpModel}. Figure \ref{fig:Bates_inf_act} shows the resulting strategy $\pi^1$ for this model.

\begin{figure}[hbt!]
	\centering
	\includegraphics[scale=0.2]{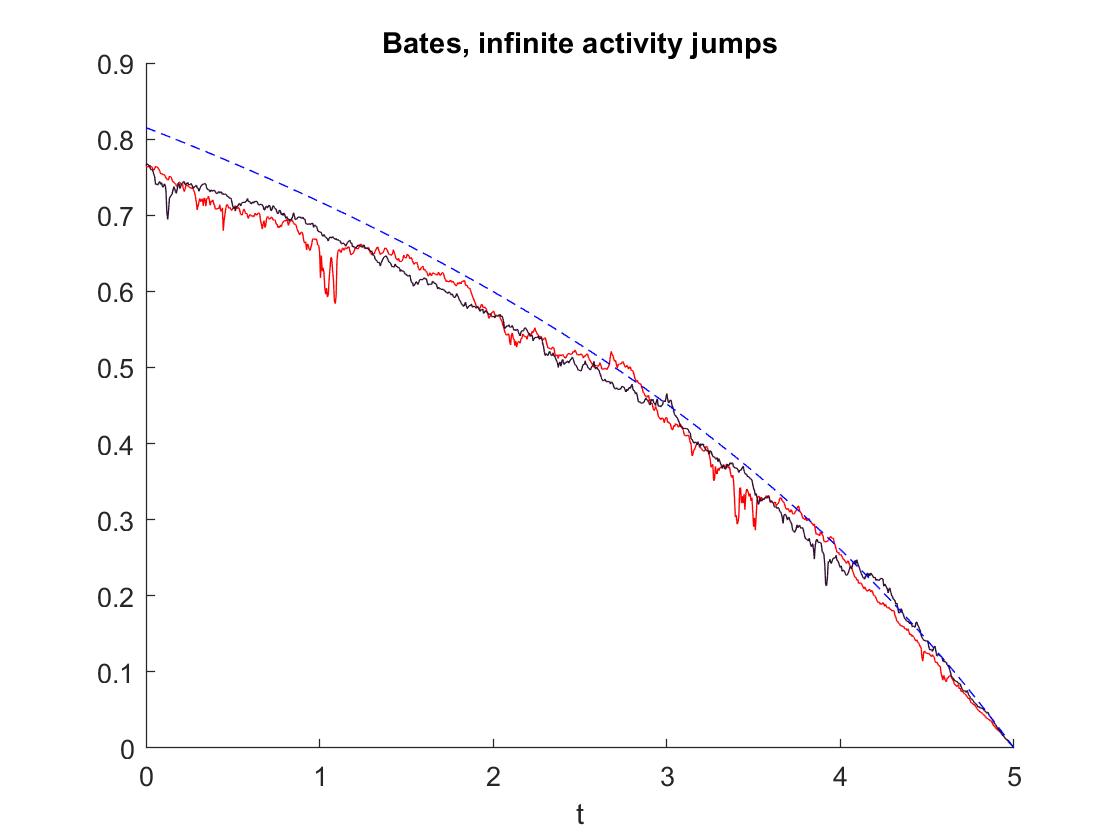}
	\caption{\footnotesize Two samples of the strategy $\pi^1$ using a time discretization of $N=1000$ steps and a space discretization of $200$ steps for Matlab's pde solver \texttt{pdepe}. The reference solution for the constant volatility (setting $z\equiv \theta$) following \cite{KornWilmott2002} is plotted in blue dashes.}
	\label{fig:Bates_inf_act}
\end{figure} 
%\FloatBarrier

\item\label{it:numb} Constant activity jumps: $\vartheta = \delta_q$. This model with constant jump size $q$ (set to $l^L_{\max}$) uses the same CIR process as in the model with infinite jumps, the only difference is - according to \ref{subsec:jumpModel} - the market price of risk
\begin{align*}
\lambda^2(z)=z\alpha_2-\frac{q}{1-\alpha_2q},
\end{align*}
where $\alpha^2=\alpha^1=2.5$.

Figure \ref{fig:Bates_const} shows the resulting strategy $\pi^2$ for this model.

\begin{figure}[hbt!]

\includegraphics[scale=0.2]{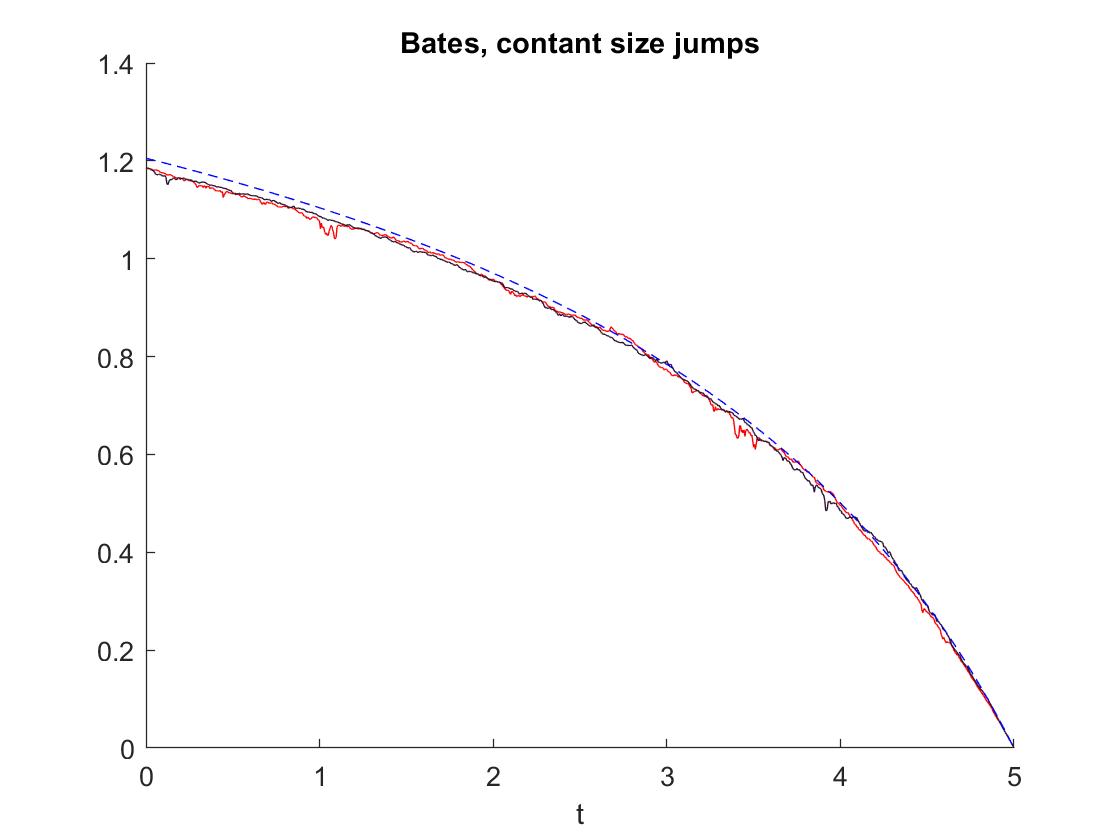}
\caption{\footnotesize Two samples of the strategy $\pi^2$ using a time discretization of $N=1000$ steps and a space discretization of $200$ steps for Matlab's pde solver \texttt{pdepe}. The reference solution for the constant volatility (setting $z\equiv \theta$) following \cite{KornWilmott2002} is plotted in blue dashes.}
\label{fig:Bates_const}
\end{figure} 
%\FloatBarrier

\item\label{numc} Absence of jumps: For the Bates model without jumps we are practically in a Heston setting. Our coefficients remain the same, except for the excess return, which just takes the form $\lambda^3(z)=\alpha_3 z$, with $\alpha_3=\alpha_1=2.5$. We obtain the following strategies $\pi^3$ for this model in Figure \ref{fig:Bates_no_jumps}.

\begin{figure}[hbt!]

\includegraphics[scale=0.2]{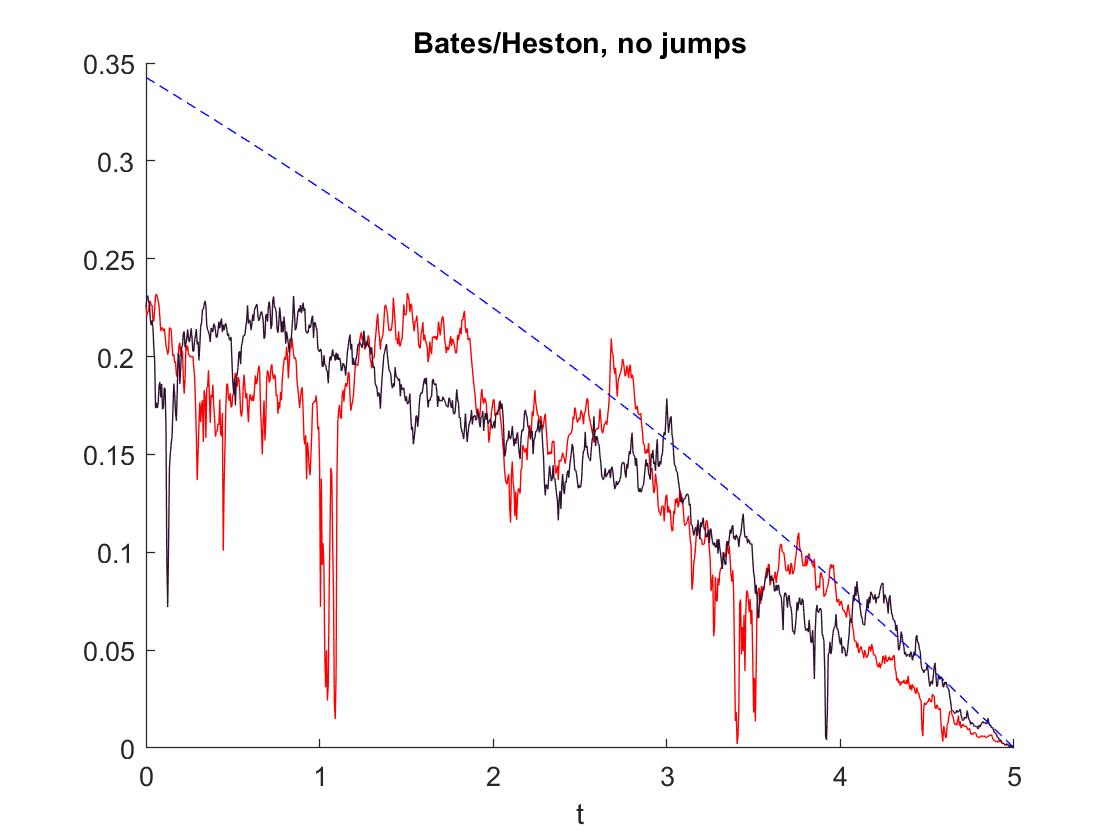}
\caption{\footnotesize Two samples of the strategy $\pi^3$ using a time discretization of $N=1000$ steps and a space discretization of $200$ steps for Matlab's pde solver \texttt{pdepe}. The reference solution for the constant volatility (setting $z\equiv \theta$) following \cite{KornWilmott2002} is plotted in blue dashes.}
\label{fig:Bates_no_jumps}
\end{figure} 
%\FloatBarrier

\item\label{it:numd}

We continue with the 'original' Heston model, assuming the same CIR-process as before, but keeping $\lambda^4$ constant at $\lambda^4\equiv\alpha^4\theta$ with $\alpha^4=\alpha^1=2.5$. In this case, the market price of risk is not appropriate, which results in a non-constant post-crash-strategy $\pi^{M,4}=\frac{\lambda^4}{(\sigma^4)^2\vee 0}$ which we illustrate in an additional figure. In those samples, one can see that the condition $\pi^4\leq \pi^{M,4}$ from Theorem~\ref{Prop:OptimalityPropForLargeMerton} is violated, so we cannot guarantee that $\pi^4$ is an actual optimal pre-crash strategy.

\begin{figure}[hbt!]

\includegraphics[scale=0.2]{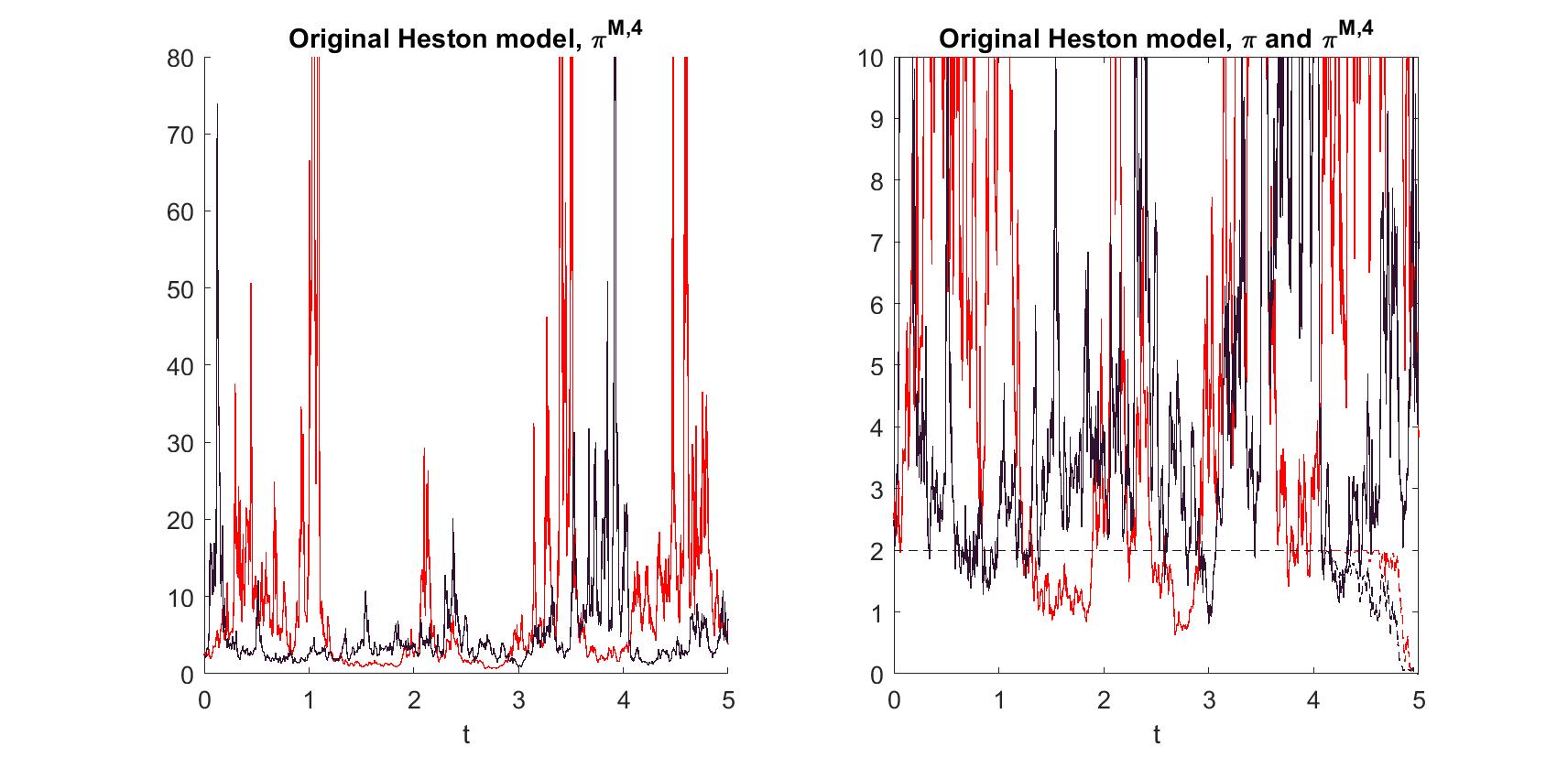}
\caption{\footnotesize Left: The post-crash optimal strategy $\pi^{M,4}$ for two sample paths. The peaks mount up to a value of about 1200. Right: The same sample graphs of $\pi^{M,4}$ together with the according strategy samples $\pi^4$ (dashed). Note that $\pi^4\nleq \pi^{M,4}$.}
\label{fig:pim4}
\end{figure} 
%\FloatBarrier

\begin{figure}[hbt!]

\includegraphics[scale=0.2]{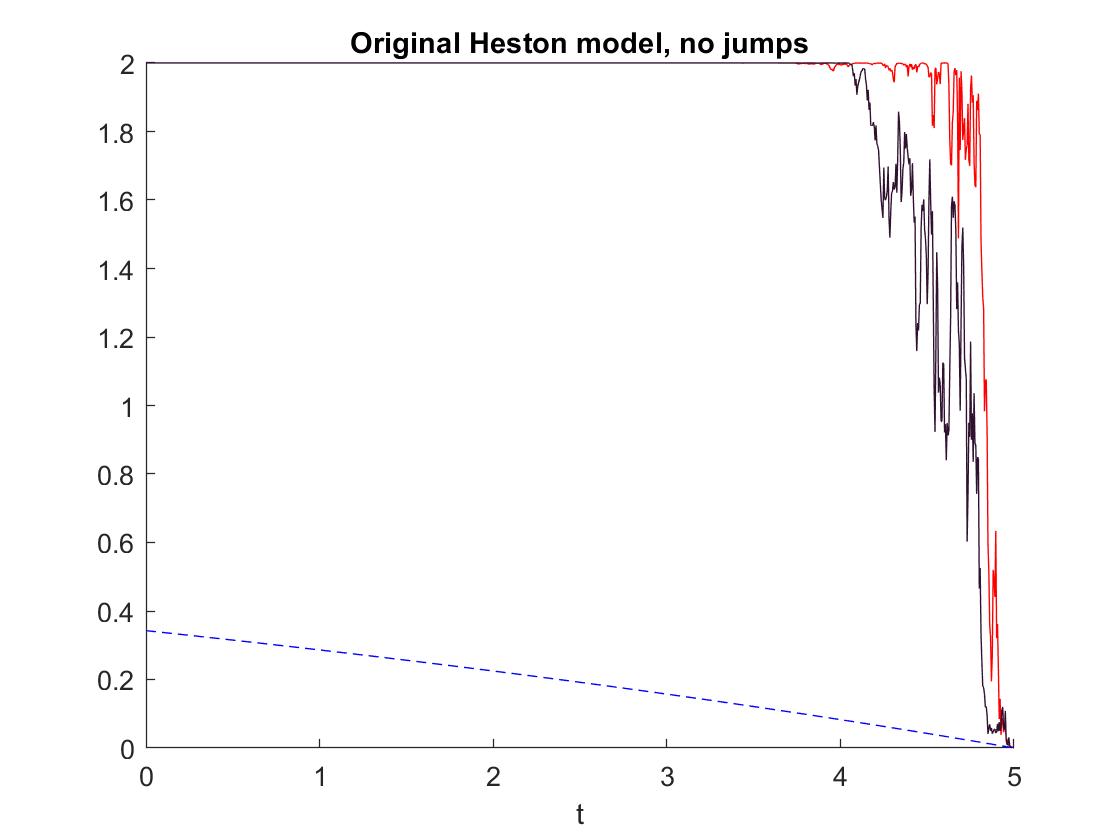}
\caption{\footnotesize Two samples of the strategy $\pi^4$ using a time discretization of $N=1000$ steps and a space discretization of $200$ steps for Matlab's pde solver \texttt{pdepe}. The reference solution for the constant volatility (setting $z\equiv \theta$) following \cite{KornWilmott2002} is plotted in blue dashes. Note that $\pi^4$ is quite close to the upper bound for admissibility, $2$, here (numerically indistinguishable for the first 3 time steps even).}
\label{fig:Heston}
\end{figure} 
%\FloatBarrier

\end{enumerate}

\subsection{Numerics for the Kim-Omberg Model}	The second model's numerical simulations that we present here are from the Kim-Omberg model. Here, the process $\lambda^5\equiv z$ is given by the Ornstein-Uhlenbeck dynamics
\begin{align*}
dz_t =\kappa^{KO}(\theta^{KO}-z_t)dt + \tilde{\varsigma}^{KO}d\hat{W}_t,\quad z_0=\theta^{KO},\quad t\in [0,T],
\end{align*}
with $\kappa^{KO}=3.5, \theta^{KO}=\theta=0.014, \tilde{\varsigma^{KO}}=0.3$,
and $\sigma^{5}\equiv\sqrt{\theta^{KO}}$.

Also in this model, we can not guarantee the condition $\pi^5\leq \pi^{M,5}$ for Theorem~\ref{Prop:OptimalityPropForLargeMerton}.

\begin{figure}[hbt!]

\includegraphics[scale=0.2]{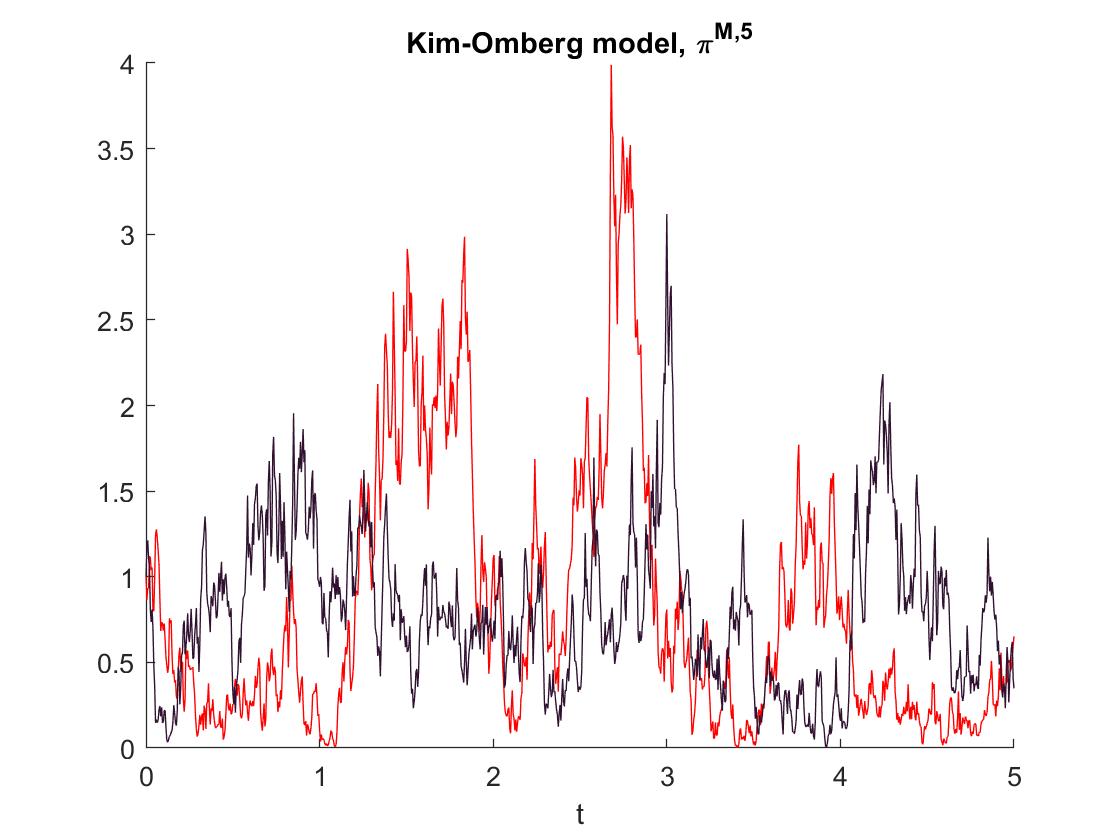}
\caption{\footnotesize The post-crash optimal strategy $\pi^{M,5}$ for two sample paths.}
\label{fig:pim5}
\end{figure} 
%\FloatBarrier

\begin{figure}[hbt!]

\includegraphics[scale=0.2]{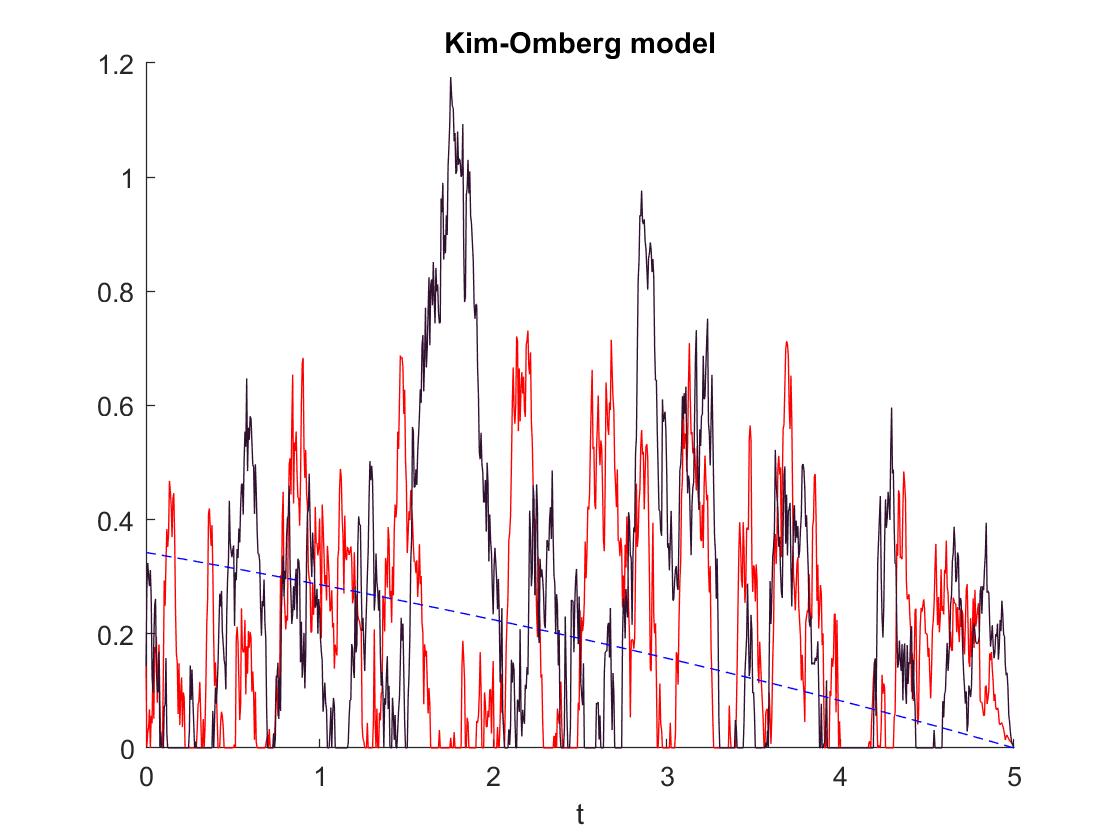}
\caption{\footnotesize Two samples of the strategy $\pi^5$ using a time discretization of $N=1000$ steps and a space discretization of $200$ steps for Matlab's pde solver \texttt{pdepe}. The reference solution (setting $z\equiv \theta$) following \cite{KornWilmott2002} is plotted in blue dashes.}
\label{fig:KO}
\end{figure} 
%\FloatBarrier

			\section*{Acknowledgements}
			We wish to thank the participants of the Stochastic Models and Control Workshop 2017 in Trier and from the London Mathematical Finance Seminar Series for useful comments and discussions.
			
			%\bibliographystyle{abbrv}
			%\bibliography{litlit}
			%\bibliography{._literature}
			\bibliographystyle{siam}
			\bibliography{litlit}
			\appendix
			
			\section{Proofs from Section~\ref{sec:post}}\label{app:post}
			
			\begin{proof}(Proposition~\ref{prop: psicont})
				There are three possibilities for $\pi^M$. Either $\pi^M=0$ or $\pi^M=\frac{1}{l^{L}_{\max}}$ or $\pi^M\in \left(0,\frac{1}{l^{L}_{\max}}\right)$.
				
				Define $\Psi(\lambda,\sigma,y):=\partial_y\Phi(y)=\lambda-\sigma^2y-\int_{[0,l^{L}_{\max}]}\frac{l}{1-yl}\vartheta(dl).$
				In the case $\pi^M\in \left(0,\frac{1}{l^{L}_{\max}}\right)$, $\Psi(\lambda,\sigma,\pi^M)=0$ and by the implicit function theorem, as $\partial_y\Psi=-\sigma^2-\int_{[0,l^{L}_{\max}]}\frac{l^2}{(1-yl)^2}\theta(dl)<0$, $\pi^M$ may be expressed by a functional relation $\pi^M=\psi(\lambda,\sigma)$, where $\psi$ is differentiable in a neighborhood around $(\lambda,\sigma)$. This shows continuity in this case. To find first Lipschitz-like relations, we start by differentiating
				$0=\partial_\lambda\Psi(\lambda,\sigma,\psi(\lambda,\sigma))$,
				and obtain, using the chain rule,
				$$\partial_\lambda\psi(\lambda,\sigma)=\frac{1}{\sigma^2+\int_{[0,l^{L}_{\max}]}\frac{l^2}{(1-\psi(\lambda,\sigma)l)^2}\theta(dl)}<\frac{1}{\sigma^2}<\infty.$$
				The same works for the derivative in direction $\sigma$,
				\begin{align*}
					&0=\partial_\sigma\Psi(\lambda,\sigma,\psi(\lambda,\sigma))\quad\Leftrightarrow\\
					&0=-2\sigma\psi(\lambda,\sigma)-\left(\sigma^2+\int_{[0,l^{L}_{\max}]}\frac{l^2}{(1-\psi(\lambda,\sigma)l)^2}\theta(dl)\right)\partial_\sigma\psi(\lambda,\sigma)\quad\Leftrightarrow\\
					&-\frac{2\sigma\psi(\lambda,\sigma)}{\sigma^2+\int_{[0,l^{L}_{\max}]}\frac{l^2}{(1-\psi(\lambda,\sigma)l)^2}\theta(dl)}=\partial_\sigma\psi(\lambda,\sigma).
				\end{align*}
				and hence, $|\partial_\sigma\psi(\lambda,\sigma)|\leq \frac{2}{l^{L}_{\max}|\sigma|}\leq \frac{2}{l^{L}_{\max}\inf_{\tilde\sigma} |\tilde\sigma|}$.
				Trivially, also in the remaining cases, $\pi^M\in \{0,\tfrac{1}{l^{L}_{\max}}\}$, it can be expressed as a function $\psi$ in $(\lambda,\sigma)$. We show that $\psi$ is continuous, no matter what case we are in: Let $(\lambda_n,\sigma_n)$ be a sequence converging to some $(\lambda,\sigma)$. Let $p^*$ be a limit point of a converging subsequence of $\psi(\lambda_{n_k},\sigma_{n_k})$ (which exists, as $\psi$ only takes values in $\left[0,\tfrac{1}{l^{L}_{\max}}\right]$). For all $k\geq 1$ is then $$\Phi(\psi(\lambda_{n_k},\sigma_{n_k}))=\Phi(\lambda_{n_k},\sigma_{n_k},\psi(\lambda_{n_k},\sigma_{n_k}))\geq \Phi(\lambda_{n_k},\sigma_{n_k},\psi(\lambda,\sigma)).$$ Since $\Phi$ is continuous in $\lambda,\sigma,y$, taking the limit $k\to\infty$ in the last inequality, we get 
				$$\Phi(\lambda,\sigma,p^*)\geq \Phi(\lambda,\sigma,\psi(\lambda,\sigma)).$$
				By the uniqueness of the $\mathrm{argmax}$ (follows from the strict convexity of $\Phi$ in $y$), we get $p^*=\psi(\lambda,\sigma).$
				
				We show that the Lipschitz-like properties hold independently from the cases: Doing this for $\lambda$ in the sequel, the proof works the similarly for $\sigma$. To that end, fix $\sigma$ and let $\lambda<\lambda'$. For an exemplary case (other cases work similar), let $\psi(\lambda,\sigma)=0$ and let $\psi(\lambda',\sigma)\in \left(0,\tfrac{1}{l^{L}_{\max}}\right)$. Then set $\lambda_1:=\inf\{t\in(\lambda,\lambda']:\psi(t,\sigma)>0\}$. In the same way set $\lambda_2:=\inf\big\{t\in (\lambda_1,\lambda'):\psi(t,\sigma)= \tfrac{1}{l^{L}_{\max}}\big\}$, and $\lambda_3:=\inf\big\{t\in [\lambda_2,\lambda'):\psi(t,\sigma)< \tfrac{1}{l^{L}_{\max}}\big\}$ and set $\lambda_2:=\lambda_3:=\lambda_1$ whenever the sets are empty. We observe,
\begin{align*}
					|\psi(\lambda,\sigma)-\psi(\lambda',\sigma)|\leq& |\psi(\lambda,\sigma)-\psi(\lambda_1,\sigma)|+|\psi(\lambda_1,\sigma)-\psi(\lambda_2,\sigma)|\\
					&+|\psi(\lambda_2,\sigma)-\psi(\lambda_3,\sigma)|+|\psi(\lambda_3,\sigma)-\psi(\lambda',\sigma)|\\
					&=|\psi(\lambda_1,\sigma)-\psi(\lambda_2,\sigma)|+|\psi(\lambda_3,\sigma)-\psi(\lambda',\sigma)|.
\end{align*}
Thus, in the intervals $(\lambda_1,\lambda_2)$ as well as $(\lambda_3,\lambda')$, $\psi$ is a differentiable function in $\lambda$ with Lipschitz constant $\frac{1}{\sigma^2}$. By taking limits, this Lipschitz property extends to the intervals' boundaries, and we obtain
				\begin{align*}
					|\psi(\lambda,\sigma)-\psi(\lambda',\sigma)|\leq\frac{1}{\sigma^2}|\lambda_1-\lambda_2|+|\lambda_3-\lambda'|\leq \frac{1}{\sigma^2}|\lambda-\lambda'|.
				\end{align*}
If $|\sigma|\geq|\tilde{\sigma}|>0$, Lipschitz continuity on the whole domain follows.
				
			\end{proof}
			
			\section{Proofs of BSDE results from Section~\ref{sec:indifferenceBSDE}} \label{App:Proofs}
			
			\begin{proof}(Proposition~\ref{Prop:ExistenceOfIndifferenceBSDESolutionJumps2})
						In view of Subsection \ref{rem:BSDEINdiff}, the generator of the BSDE can be written as 
						\begin{align*}
							&{f}(t,y)=\\
							&a_t+\frac{\lambda_t^-}{l^{WOC}}(1-e^{-(y\vee 0)})-\frac{\lambda_t^+}{l^{WOC}}(1-e^{-(y\vee 0)})+\frac{\sigma_t^2}{2\left(l^{WOC}\right)^2}(1-e^{-(y\vee 0)})^2\\
							&-\int_{[0,l^{L}_{\max}]}\log\left(1-\frac{1-e^{-(y\vee 0)}}{l^{WOC}}l\right)\vartheta(dl),
						\end{align*}
						where $a_t$ is given by $$\lambda_t\pi_t^M-\frac{\sigma_t^2(\pi_t^M)^2}{2}+\int_{[0,l^{L}_{\max}]}\log\left(1-\pi_t^Ml\right)\vartheta(dl).$$
						Now, given the boundedness conditions on $\vartheta, \lambda, \sigma, \pi^M$, the generator $f$ meets the assumptions of \cite[Theorem 3.2]{geiss2018monotonic}(or e.g.~\cite[Theorem 3.3]{KremsnerSteinicke}, \cite[Theorem 1]{Kruse}), as seen in Subsection \ref{rem:BSDEINdiff}, \eqref{eq:fcondI} and \eqref{eq:fcondLip}, from which the assertion follows.
					\end{proof}
					
					\begin{proof}(Theorem~\ref{Thm:ExistenceOfIndifferenceBSDESolutionUnbounded2Jumps})
						We approximate the generator ${f}$ by the functions
						\begin{align*}
							f^{n,N}(t,y):=&a^{N,n}_t+\frac{\lambda_t^-\wedge n}{l^{WOC}}(1-e^{-(y\vee 0)})-\frac{\lambda_t^+}{l^{WOC}}(1-e^{-(y\vee 0)})\\
							&+\frac{\sigma_t^2\wedge n}{2\left(l^{WOC}\right)^2}(1-e^{-(y\vee 0)})^2-\int_{[0,l^{L}_{\max}]}\log\left(1-\frac{1-e^{-(y\vee 0)}}{l^{WOC}}l\right)\vartheta(dl),
						\end{align*}
						with $a^{N,n}_t:=(\lambda_t\wedge n)\pi_t^M-\frac{(\sigma_t^2\wedge N)(\pi_t^M)^2}{2}+\int_{[0,l^{L}_{\max}]}\log\left(1-\pi_t^Ml\right)\vartheta(dl)$, and
					\begin{align*}
							f^{N}(t,y):=&a^{N}_t+\frac{\lambda_t^-}{l^{WOC}}(1-e^{-(y\vee 0)})-\frac{\lambda_t^+}{l^{WOC}}(1-e^{-(y\vee 0)})\\
							&+\frac{\sigma_t^2}{2\left(l^{WOC}\right)^2}(1-e^{-(y\vee 0)})^2-\int_{[0,l^{L}_{\max}]}\log\left(1-\frac{1-e^{-(y\vee 0)}}{l^{WOC}}l\right)\vartheta(dl),
						\end{align*}
						with $a^{N}_t:=\lambda_t\pi_t^M-\frac{(\sigma_t^2\wedge N)(\pi_t^M)^2}{2}+\int_{[0,l^{L}_{\max}]}\log\left(1-\pi_t^Ml\right)\vartheta(dl)$.
						
						For the BSDEs with generators $f^{N,n}$ (and terminal condition $0$), we get solutions $(\Upsilon^{N,n},\sigma_\Upsilon^{N,n})$ from Proposition \ref{Prop:ExistenceOfIndifferenceBSDESolutionJumps2}. Note that $f^{N,n}(t,\omega,y)\nearrow {f^N}(t,\omega,y)$ and ${f^N}(t,\omega,y)\searrow f(t,\omega,y)$. We will follow this order of approximations to construct solutions. For the first approximation, by the comparison theorem (e.g.~in \cite[Theorem 3.4]{geiss2018monotonic}, \cite[Theorem 6.1]{KremsnerSteinicke}, \cite[Proposition 4]{Kruse}), we get that $\left(\Upsilon_t^{N,n}(\omega)\right)_{n\geq 1}$ is a monotonically increasing family for all $t\in [0,T]$ for almost all $\omega\in\Omega$. Thus, we can define the random variable $\Upsilon^N_t(\omega):=\lim_{n\to\infty}\Upsilon_t^{N,n}(\omega)$, wherever the limit exists and $0$ on the complement.

						We show that $\mathbb{E}\left[\sup_{t}|\Upsilon^N_t|^2\right]<\infty$. To that end, observe first that It\^o's formula yields for all $t\in [0,T]$
						\begin{align}\label{Thm:ExistenceOfIndifferenceBSDESolutionUnbounded2JumpsItou}
							|\Upsilon_t^{N,n}|^2+\int_t^T|\sigma^{N,n}_{\Upsilon,s}|^2ds&+\int_0^T\int_{[0,l^{L}_{\max}]}|U^{N,n}_{\Upsilon,s}(l)|^2\vartheta(dl)ds\nonumber\\
							&=\int_t^T2\Upsilon^{N,n}_sf^{N,n}(s,\Upsilon^{N,n}_s)ds-\int_t^T2\Upsilon_s^{N,n}\sigma^{N,n}_{\Upsilon,s}d{\bar W}_s\nonumber\\
							&\quad-\int_{(t,T]\times [0,l^{L}_{\max}]}\left(\left(\Upsilon^{N,n}_{s-}+U^{N,n}_{\Upsilon,s}\right)^2-\left(\Upsilon^{N,n}_{s-}(l)\right)^2\right)\tilde{\nu}(ds,dl).
						\end{align}
						From this equation, we get by taking conditional expectations w.r.t.~$\mathcal{F}_t$ and using Young's inequality that for all $q\in(0,\infty)$,
						\begin{align}\label{Thm:ExistenceOfIndifferenceBSDESolutionUnbounded2JumpsItouExp}
							&\mathbb{E}\left[\int_0^T|\sigma^{N,n}_{\Upsilon,s}|^2ds\right]+\mathbb{E}\left[\int_0^T\int_{[0,l^{L}_{\max}]}|U^{N,n}_{\Upsilon,s}(l)|^2\vartheta(dl)ds\right]\\
							&\quad\leq \frac{1}{q}\mathbb{E}\left[\sup_{t}|\Upsilon^{N,n}_t|^2\right]+q\mathbb{E}\left[\int_0^Tf^{N,n}(s,\Upsilon^{N,n}_s)ds\right]^2.\nonumber
						\end{align}
						Further, \eqref{Thm:ExistenceOfIndifferenceBSDESolutionUnbounded2JumpsItou} implies, together with Doob's martingale inequality, Burkholder-Davis-Gundy's, Young's inequality and the fact that $$\left|\left(\Upsilon^{N,n}_{s-}+U^{N,n}_{\Upsilon,s}(l)\right)^2-\left(\Upsilon^{N,n}_{s-}\right)^2\right|\leq 4\sup_t|\Upsilon^{N,n}_t||U^{N,n}_{\Upsilon,s}(l)|,$$ that there is a constant $c$, such that for all $R\in(0,\infty)$,
						\begin{align*}&\mathbb{E}\left[\sup_{t}|\Upsilon^{N,n}_t|^2\right]\leq \frac{1}{R}\mathbb{E}\left[\sup_{t}|\Upsilon^{N,n}_t|^2\right]+R\mathbb{E}\left[\int_0^Tf^{N,n}(s,\Upsilon^{N,n}_s)ds\right]^2\\
							&\quad\quad\quad+4c\mathbb{E}\left[\int_0^T\left|\Upsilon^{N,n}_s\sigma^{N,n}_{\Upsilon,s}\right|^2ds\right]^\frac{1}{2}+8c\mathbb{E}\left[\int_0^T\sup_t|\Upsilon^{N,n}_t|^2\left|U^{N,n}_{\Upsilon,s}(l)\right|^2\vartheta(dl)ds\right]^\frac{1}{2}\\
							&\leq\frac{1}{R}\mathbb{E}\left[\sup_{t}|\Upsilon^{N,n}_t|^2\right]+R\mathbb{E}\left[\int_0^Tf^{N,n}(s,\Upsilon^{N,n}_s)ds\right]^2\\
							&\quad\quad+\frac{8c}{R}\mathbb{E}\left[\sup_{t}|\Upsilon^{N,n}_t|^2\right]+cR\mathbb{E}\left[\int_0^T|\sigma^{N,n}_{\Upsilon,s}|^2ds\right]+cR\mathbb{E}\left[\int_0^T\left|U^{N,n}_{\Upsilon,s}(l)\right|^2\vartheta(dl)ds\right],
						\end{align*}
						where we used Young's inequality again for the second estimate.
						Replacing the last term with the help of inequality \eqref{Thm:ExistenceOfIndifferenceBSDESolutionUnbounded2JumpsItouExp},
						we arrive at
						\begin{align*}\mathbb{E}\left[\sup_{t}|\Upsilon^{N,n}_t|^2\right]\leq&\left(\frac{1+16c}{R}+\frac{cR}{q}\right)\mathbb{E}\left[\sup_{t}|\Upsilon^{N,n}_t|^2\right]\\
						&+\left(R+cRq\right)\mathbb{E}\left[\int_0^Tf^{N,n}(s,\Upsilon^{N,n}_s)ds\right]^2.
						\end{align*}
						Choosing now $R$ and $q$ such that $\frac{1+16c}{R}+\frac{cR}{q}<1$, we find a constant $C\in(0,\infty)$ such that
						$\mathbb{E}\left[\sup_{t}|\Upsilon^{N,n}_t|^2\right]\leq C\mathbb{E}\left[\int_0^Tf^{N,n}(s,\Upsilon^{N,n}_s)ds\right]^2$. From our integrability assumption on $\lambda$ and $\sigma$ follow the uniform integrability of $\left(\sup_{t}|\Upsilon^{N,n}_t|^2\right)_{n\geq 1}$ and thus also $\mathbb{E}\left[\sup_{t}|\Upsilon^N_t|^2\right]<\infty$.

						Now, by dominated convergence we get that for $t\in [0,T]$,
						$$\lim_{n\to\infty}\int_t^Tf^{N,n}(s,\Upsilon_s^{N,n})ds=\int_t^T f^N(s,\Upsilon^N_s)ds,\quad\text{a.s.}$$
						and hence 
						\begin{align*}
							\Upsilon^N_t-\int_t^Tf^N(s,\Upsilon^N_s)ds&=\lim_{n\to\infty}\left(\Upsilon_t^{N,n}-\int_t^Tf^{N,n}(s,\Upsilon_s^{N,n})ds\right)\\
							&=\lim_{n\to\infty}\left(-\int_t^T\sigma^{N,n}_{\Upsilon,s}d{\bar W}_s-\int_{]t,T]\times [0,l^{L}_{\max}]}U^{N,n}_{\Upsilon,s}(l)\tilde{\nu}(ds,dl)\right).
						\end{align*}
						Together with the integrability assumptions on $\lambda,\sigma$ and since for all $t\in [0,T]$, $\mathbb{E}\left[|\Upsilon^N_t|^2\right]<\infty$, it follows that $\int_t^T\sigma^{N,n}_{\Upsilon,s}d{\bar W}_s+\int_{]t,T]\times [0,l^{L}_{\max}]}U^{N,n}_{\Upsilon,s}(l)\tilde{\nu}(ds,dl)$ converges in $L^2$ to a random variable $V_t\in L^2$.
						Therefore,
						$$0=\mathbb{E}\left[\int_t^T\sigma^{N,n}_{\Upsilon,s}d{\bar W}_s+\int_{]t,T]\times [0,l^{L}_{\max}]}U_{\Upsilon,s}(l)\tilde{\nu}(ds,dl)\middle|\mathcal{F}_t\right]\to\mathbb{E}\left[V_t\middle|\mathcal{F}_t\right],\quad\text{a.s.},$$
						and $\mathbb{E}\left[V_t\middle|\mathcal{F}_t\right]=0$ follows. In particular, $V_0=-\Upsilon^N_0+\int_t^Tf^N(s,\Upsilon^N_s)ds\in L^2$.
						By the martingale representation theorem, there are processes $A,B$, $A$ progressively measurable and $B$ predictable such that 
						$$V_0=\int_0^TA_sd{\bar W}_s+\int_{(0,T]\times [0,l^{L}_{\max}]}B_s(l)\tilde{\nu}(ds,dl),$$
						and
						$$\mathbb{E}\left[\int_0^TA_s^2ds\right]+\mathbb{E}\left[\int_0^T\int_0^{l^{L}_{\max}}B_s(l)^2\vartheta(dl)ds\right]<\infty.$$
						%where the integrability follows from the Burkholder-Davis-Gundy inequality (for the Brownian part) and that for a martingale $M$, 
						%$$\mathbb{E}\sup_{t\in [0,T]}|M_t|^p\leq \frac{1}{1-p}\left[\mathbb{E}|M_T|\right]^p,\quad p\in (0,1).(Cite Briand CiteRevuzYor)$$
						For all $t\in [0,T]$ we have 
						$V_0-V_t=-\Upsilon^N_0+\Upsilon^N_t+\int_0^tf^N(s,\Upsilon^N_s)ds$, which is $\mathcal{F}_t$-measurable implying $\mathbb{E}\left[V_0-V_t\middle|\mathcal{F}_t\right]=V_0-V_t$.
						Altogether, we have that 
						\begin{align*}
							V_0-V_t&=\mathbb{E}\left[V_0-V_t\middle|\mathcal{F}_t\right]=\mathbb{E}\left[V_0\middle|\mathcal{F}_t\right]-\mathbb{E}\left[V_t\middle|\mathcal{F}_t\right]\\
							&=\int_0^tA_sd{\bar W}_s+\int_{]0,t]\times [0,l^{L}_{\max}]}B_s(l)\tilde{\nu}(ds,dl)-0.
						\end{align*}
						As a consequence,
						\begin{align*}
							V_0-V_t&=\int_0^TA_sd{\bar W}_s+\int_{]0,T]\times [0,l^{L}_{\max}]}B_s(l)\tilde{\nu}(ds,dl)-V_t\\
							&=\int_0^tA_sd{\bar W}_s+\int_{]0,t]\times [0,l^{L}_{\max}]}B_s(l)\tilde{\nu}(ds,dl),
						\end{align*}
						from which we infer $V_t=\int_t^T A_sd{\bar W}_s+\int_{]t,T]\times [0,l^{L}_{\max}]}B_s(l)\tilde{\nu}(ds,dl)$. Now it is readily checked that $(\Upsilon,A,B)$ solves the BSDE given by ${f^N}$, so we can set $\sigma^N_\Upsilon:=A$ and $U^N_\Upsilon:=B$. We go over to the approximation in $N$ now. Note that for all $N\geq M$, $n\geq 0$ and all $t\in [0,T]$ we have that $\mathbb{P}$-a.s.
						\begin{align*}
						\Upsilon^{N,n}_t\geq \Upsilon^{M,n}_t,\quad\text{hence}\quad \Upsilon^N_t\geq \Upsilon^M_t.
						\end{align*}
						Therefore, we may define again for all $t$, $\Upsilon_t:=\lim_{N\to\infty}\Upsilon^N_t$ on the set where the limit exists and 0 otherwise. Now the same steps from \eqref{Thm:ExistenceOfIndifferenceBSDESolutionUnbounded2JumpsItou} can be performed again, ending up with a solution $(\Upsilon,\sigma_\Upsilon,U_\Upsilon)$.
						
						%
						%Now, taking limits and using the comparison theorem, the proof can be concluded in the same way as the one of Theorem \ref{Thm:ExistenceOfIndifferenceBSDESolutionUnboundedJumps}.
					\end{proof}

					\begin{proof}(Theorem~\ref{Thm:UniquenessOfIndifferenceBSDESolutionUnboundedJumps})
						Existence follows by Theorem \ref{Thm:ExistenceOfIndifferenceBSDESolutionUnbounded2Jumps} and the variants depicted in Remark \ref{rem:BSDErem} as $\lambda$ and $\sigma^2$ allow finite moments of any order $p>0$. 
						Denoting differences between two supposed solutions $\Upsilon,\Upsilon'$ by $\Delta\Upsilon:=\Upsilon-\Upsilon'$, we get, using the It\^o formula given in \cite[Lemma 7]{Kruse} for $p=1$, 
						\begin{align*}
							\mathbb{E}\left[|\Delta \Upsilon_t|\right]\leq \mathbb{E}\left[\int_t^T|f(s,\Upsilon_s)-f(s,\Upsilon'_s)|ds\right].
						\end{align*}
						
						We now split up the set $\Omega$ into $C_n(t)=\left\{\omega:\frac{\lambda_t^-(\omega)}{l^{WOC}}+\frac{\sigma_t^2}{2(l^{WOC})^2}\leq n\right\}$ and $\Omega\setminus C_n(t)$ to estimate, 
						using Lipschitz and boundedness properties of the function $y\mapsto 1-\exp(-(y\vee 0))$, 
						\begin{align*}
							\mathbb{E}\left[|\Delta \Upsilon_t|\right]\leq n\int_t^T \mathbb{E}\left[\Delta|\Upsilon_s|\right]ds+c\mathbb{E}\left[\int_t^T\chi_{\Omega\setminus C_n(s)}2\left(\lambda_s^-+\sigma_s^2\right)ds\right],
						\end{align*}
						where $c=\left(\frac{1}{l^{WOC}}\vee\frac{1}{2(l^{WOC})^2}\right)$.
						Gronwall's inequality now shows that for $r\in [t,T]$:
						\begin{align*}
							\mathbb{E}\left[|\Delta \Upsilon_r|\right]\leq e^{n(T-t)}c\mathbb{E}\left[\int_t^T\chi_{\Omega\setminus C_n(s)}2\left(\lambda_s^-+\sigma_s^2\right)ds\right].
						\end{align*}
						The terms can be further estimated, using the definition of $C_n$ and that $\frac{e^{a x}}{a}\geq x$, by
						\begin{align*}
							\mathbb{E}\left[|\Delta \Upsilon_r|\right]\leq c\mathbb{E}\left[\int_t^T\chi_{\Omega\setminus C_n(s)}\frac{e^{4(T-t)\max\{1/c,1\}\left(\lambda_s^-+\sigma_s^2\right)}}{4(T-t)\max\{1/c,1\}}ds\right].
						\end{align*}
						If $t_1$ is such that $4(T-t_1)\max\{1/c,1\}<\varepsilon$, then the right hand side tends to zero as $n\to\infty$, showing that $\Delta\Upsilon=0$ on $[t_1,T]$. We can now perform the same steps as above for the BSDE 
						$$\Delta\Upsilon_t=\int_t^{t_1}\left(f(s,\Upsilon_s)-f(s,\Upsilon'_s)\right)ds-\int_t^{t_1}\sigma_{\Upsilon,s}d\bar{W}_s-\int_{(t,t_1]\times[0,l^{L}_{\max}]}U_{\Upsilon,s}(l)\tilde{\nu}(ds,dl),$$
						$t\in [0,t_1]$, showing iteratively that $\Delta\Upsilon=0$ on the whole interval $[0,T]$. Uniqueness of $\sigma_{\Upsilon}$ and $U_\Upsilon$ then follow.
					\end{proof}
					
					\begin{proof}(Theorem~\ref{Thm:ComparisonOfIndifferenceBSDESolutionUnboundedJumps})
						Denoting differences $\Upsilon-\Upsilon',\sigma_\Upsilon-\sigma_\Upsilon',U-U'$ by $\Delta \Upsilon, \Delta \sigma_\Upsilon, \Delta U_\Upsilon$, and letting $B(s):=\left\{l\in[0,l^{L}_{\max}]: \Delta U_{\Upsilon,s}(l)\geq-\Delta \Upsilon_s\right\}$, by Tanaka-Meyer's formula we get that
						\begin{align*}
							&((\Delta \Upsilon_t)^+)^2=\ \xi-\xi'+\int_t^T\chi_{\left\{\Delta \Upsilon_s \geq 0\right\}}\bigg[(\Delta \Upsilon_s)\left(f(s,\Upsilon_s)-f'(s,\Upsilon'_s)\right)-(\Delta \sigma_{\Upsilon,s})^2\\
							&-\int_t^T\int_{B(s)}(\Delta U_{\Upsilon,s}(l))^2\vartheta(dl)+\int_{B(s)^c}\left((\Delta \Upsilon_s)^2+2(\Delta U_{\Upsilon,s}(l))(\Delta \Upsilon_s)\right)\vartheta\bigg]ds+M(t),
						\end{align*}
						where $M(t)$ is a stochastic integral term with zero expectation. Thus, omitting negative terms, we get
						\begin{align*}
							\mathbb{E}\left[((\Delta \Upsilon_t)^+)^2\right]=\mathbb{E}\left[\int_t^T\chi_{\left\{\Delta \Upsilon_s \geq 0\right\}}(\Delta \Upsilon_s)\left(f(s,\Upsilon_s)-f'(s,\Upsilon'_s)\right)ds\right].
						\end{align*}
						We add and subtract $f(s,\Upsilon'_s)$ in the integral and get, using that $f(s,\Upsilon'_s)\leq f'(s,\Upsilon'_s)$,
						\begin{align*}
							\mathbb{E}\left[((\Delta \Upsilon_t)^+)^2\right]=\mathbb{E}\left[\int_t^T\chi_{\left\{\Delta \Upsilon_s \geq 0\right\}}(\Delta \Upsilon_s)\left(f(s,\Upsilon_s)-f(s,\Upsilon'_s)\right)ds\right]
						\end{align*}
						(this step can also be performed with inserting $f'(s,\Upsilon_s)$ and using the inequality $f(s,\Upsilon'_s)\leq f(s,\Upsilon'_s)$, if this is the inequality assumed). The assumption on $f$ now implies
						\begin{align*}
							\mathbb{E}\left[((\Delta \Upsilon_t)^+)^2\right]=\mathbb{E}\left[\int_t^T K_s((\Delta \Upsilon_s)^+)^2\wedge K_s(\Delta \Upsilon_s)^+ds\right]
						\end{align*}

						We now split up the set $\Omega$ into $C_n(t)=\left\{\omega:K_t\leq n\right\}$ and $\Omega\setminus C_n(t)$ to estimate
						\begin{align*}
							\mathbb{E}\left[((\Delta \Upsilon_t)^+)^2\right]\leq &\ n\int_t^T \mathbb{E}\left[((\Delta \Upsilon_s)^+)^2\right]ds\\
							&+\mathbb{E}\left[\int_t^T\chi_{\Omega\setminus C_n(s)}2K_s^2ds\right]+\mathbb{E}\int_s^T2((\Delta \Upsilon_s)^+)^2ds.
						\end{align*}
						Gronwall's inequality now shows that for $r\in [t,T]$:
						\begin{align*}
							\mathbb{E}\left[((\Delta \Upsilon_r)^+)^2\right]\leq e^{(n+2)(T-t)}\mathbb{E}\left[\int_t^T\chi_{\Omega\setminus C_n} 2K_s^2 ds\right].
						\end{align*}
						The terms can be further estimated by
						\begin{align*}
							\mathbb{E}\left[((\Delta \Upsilon_r)^+)^2\right]\leq e^{(n+2)(T-t)}\mathbb{E}\left[\int_t^T\chi_{\Omega\setminus C_n} 4e^{8(T-t)K_s} ds\right].
						\end{align*}
						If $t_1$ is such that $8(T-t_1)<\varepsilon$, then the right hand side tends to zero as $n\to\infty$, showing that $\Delta \Upsilon\leq 0$ on $[t_1,T]$. As in the proof of Theorem \ref{Thm:UniquenessOfIndifferenceBSDESolutionUnboundedJumps}, we can now show successively that on small enough intervals the process $(\Delta \Upsilon)^+=0$, from which we infer the assertion.
					\end{proof}
					\phantomsection\label{proof:propaAlmostIto}
				\begin{proof}(Proposition \ref{Prop:MertonSubindifference})
					With $\pi^M$ also $\Upsilon^M$ is almost an It\^o process, it is just extended by a term involving a local time, in particular, using Tanaka-Meyer's formula \cite[Chapter 4, Theorem 70 and Corollary 1]{Protter}, we have
					\begin{equation}\label{eq:piMsubdiff}
						\begin{aligned}
						\Upsilon^M_t &= -\log(1-l^{WOC}\pi^M_t) \\
						&=-\log(1-l^{WOC}\pi^M_0)+\frac{(l^{WOC})^2L^\varpi_t}{2} +  \int_0^t\frac{l^{WOC}{\bf 1}_{[0,\infty)}}{1-l^{WOC}\pi^M_s}b_s (\varpi_s) d\bar{W}_s\\
						&\quad +\int_0^t\bigg(\frac{l^{WOC}}{1-l^{WOC}\pi^M_s}{\bf 1}_{[0,\infty)}(\varpi_s)a_s+\frac{\left(l^{WOC}\right)^2|b_s|^2}{(1-l^{WOC}\pi^M_s)^2}{\bf 1}_{[0,\infty)}(\varpi_s)
						\bigg)ds
						\end{aligned}
					\end{equation}
					where $L^\varpi$ is the local time at $0$ of the process $\varpi$. By the assumptions on $\pi^M$, the stochastic integrals are martingales again. We name the integrands $a_\Upsilon,b_\Upsilon$. 
					Since the $\pi^M$ is post-crash optimal, $Z^M:=Z^{\pi^M}$ is just given by $Z^M = -\Upsilon^M$ and as $Z^M$ is by assumption a submartingale, $\Upsilon^M$ must be a supermartingale, that is $a_{\Upsilon,t}dt+\frac{(l^{WOC})^2}{2}dL^\varpi_t$ is a measure with values in $(-\infty,0]$. We can view $(\Upsilon^M,-b_\Upsilon)\in \mathcal{S}^2\times L^2(\bar{W})$ then as the solution to the (slightly generalized) BSDE 
					\begin{align*}
					\Upsilon^M_t=\xi-\int_t^T a_{\Upsilon,s}ds-\frac{(l^{WOC})^2(L^\varpi_T-L^\varpi_t)}{2}+\int_t^T b_{\Upsilon,s} d\bar{W}_s
					\end{align*}		
					with driver $-a_\Upsilon$, additional measure term $-\frac{(l^{WOC})^2}{2}dL^\varpi_t$ and terminal value $\Upsilon_T^M$. For such generalized BSDEs with data $(f,dR)$ ($f$ being the generator and $dR$ being an additional measure term) we may apply the comparison theorem from 
					%Theorem \ref{Thm:ComparisonOfIndifferenceBSDESolutionUnboundedJumps} 
					\cite[Proposition 1]{eddahbi_fakhouri_ouknine_2017} to compare it to $\hat{\pi}$ which is the solution to the indifference BSDE (\ref{Prop:LocalIndifferenceStrategiesWithStochCoefficients:BSDE}). Clearly $\Upsilon_T^M\geq0=\hat{\Upsilon}_T$. For comparison of the BSDEs it is sufficient to compare them only along the solution path of one of the two BSDEs. Here we choose comparison along $(\Upsilon^M,-b_\Upsilon)$. The data of the first BSDE is, as pointed out above, $(f_1,dR_1)=\left(-a_{\Upsilon,t},-\frac{(l^{WOC})^2}{2}dL^\varpi_t\right)$. Conversely, the driver of the second BSDE evaluated along the path of $\Upsilon^M_t$ is $\Phi_t(\pi_t^M)-\Phi_t(\pi_t^M) = 0$, as is its measure term. So the data of the second BSDE is $(f_2,dR_2)=(0,d0)$. By recapitulating the proof of the comparison theorem \cite[Proposition 1]{eddahbi_fakhouri_ouknine_2017}, which requires $f_2\leq f_1$ as well as $dR_2\leq dR_1$, it is easy to see (in the proof's last inequality) that also the condition $f_2ds+dR_2\leq f_1ds+dR_1$ is sufficient. Therefore, we can conclude $\hat{\Upsilon}\leq\Upsilon^M$ and as utility crash exposures are just monotone transformations of portfolio processes, this implies $\hat{\pi}\leq\pi^M$.
				\end{proof}
					
								\section{Proofs of (CIR) results from Section~\ref{sec:Markovian}} \label{App:Markov}
								
												\begin{proof}(Theorem~\ref{Thm:existenceOfViscositySolutionJumps})
					For bounded $\lambda,\sigma$ we refer to \cite[Theorem 3.5]{bbp}, since the conditions of the generator of the BSDE \eqref{Prop:LocalIndifferenceStrategiesWithStochCoefficients:BSDE}
					are met. 
					
					We therefore focus on the unbounded case and follow the proof of \cite[Theorem 3.4]{bbp}. 
					Let $\Upsilon^t_s(x)$ be defined by the solution of \eqref{Prop:LocalIndifferenceStrategiesWithStochCoefficients:BSDE}
					on $[t,T]$. 
					with $\lambda_s=\lambda(z_s^t(x))$, $\sigma_t=\sigma(z_s^t(x))$, where $z^t(x)$ is the solution to the forward equation \eqref{eq:stateSDE} starting from $z^t_t=x$ on $[t,T]$. Define then $v(t,x):=\Upsilon^t_t(x)$ 
					What we have to show to conduct the proof as in \cite{bbp} is the continuity of $v$, uniqueness of the family of BSDE-solutions $\left(\Upsilon^t(x)\right)_{(t,x)\in[0,T]\times\mathbb{R}}$, a comparison theorem  and an inequality used in their proof which we treat below. In \cite{bbp}, all those properties follow from the uniform Lipschitz condition for their generator functions, which does not hold in the case of our generator $f$ because $\lambda$ and $\sigma$ are unbounded. 
					
					The uniqueness of the solutions to the family of BSDEs in our case is granted by Theorem \ref{Thm:UniquenessOfIndifferenceBSDESolutionUnboundedJumps} and the exponential integrability condition \eqref{ass:BSDEXPB} on $\lambda(z_t)$ and $\sigma(z_t)$. The comparison theorem needed is given by Theorem \ref{Thm:ComparisonOfIndifferenceBSDESolutionUnboundedJumps}.
					
					We show the continuity of $v$:
					We define $\Upsilon^t_s(x):=\Upsilon^t_t(x)$ for $s\leq t$ and estimate
					\begin{align*}
						&|v(t,x)-v(t',x')|^2=|\Upsilon_0^t(x)-\Upsilon_0^{t'}(x')|^2\leq\mathbb{E}\left[\sup_{s\in{[0,T]}}|\Upsilon^t_s(x)-\Upsilon^{t'}_s(x')|^2\right].
					\end{align*}
					Using It\^o's formula for $|\Upsilon^t_s(x)-\Upsilon^{t'}_s(x')|^2$, and subsequently Young's inequality and Doob's maximal inequality through standard BSDE methods, we get that
					\begin{small}
					\begin{align*}
						\mathbb{E}\left[\sup_{s\in{[0,T]}}|\Upsilon^t_s(x)-\Upsilon^{t'}_s(x')|^2\right]\leq \mathbb{E}\left[\int_0^T\left|\chi_{[t,T]}(s)f(s,\Upsilon^{t}_s(x))-\chi_{[t',T]}(s)f(s,\Upsilon^{t'}_s(x'))\right|^2ds\right].
					\end{align*}
					\end{small}
					We therefore have to investigate 
					\begin{align*}
						&\int_0^T\left|\chi_{[t,T]}(s)f(s,\Upsilon^t_s(x))-\chi_{[t',T]}(s)f(s,\Upsilon^{t'}_s(x'))\right|^2ds\\
						&=\int_t^{t'}|f(s,\Upsilon_s^t(x))|^2ds+\int_{t'}^T\left|f(s,\Upsilon^t_s(x))-f(s,\Upsilon^{t'}_s(x'))\right|^2ds,
					\end{align*}
					assuming $t<t'$. Concerning the first summand, note first that whenever $$\int_{[0,l^{L}_{\max}]}\log\left(1-\frac{l}{l^{L}_{\max}}\right)\vartheta(dl)=-\infty,$$ by Case 1 of the proof of Proposition \ref{prop: argmax}, 
					\begin{align*}
						&\int_{[0,l^{L}_{\max}]}\log\left(1-\psi(\lambda(z_s(x))\sigma(z_s(x)))l\right)\vartheta(dl)\\
						&\leq \psi(\lambda(z_s(x))\sigma(z_s(x)))(\lambda(z_s(x))-\sigma(z_s(x))^2\psi(\lambda(z_s)\sigma(z_s))).
					\end{align*}
					We may then use the growth and boundedness assumptions on $\lambda,\sigma, z, \vartheta$ and the form of $\Phi$ to find a constant $C>0$ and estimate the generator by
					\begin{align*}
						|f(s,y)|^2\leq C\left(1+|z^t_s(x)|^{2p}\right).
					\end{align*}
					As there is a $K>0$ such that $$\mathbb{E}\left[\int_0^T(1+|z^t_s(x)|^{2p})ds\right]\leq \mathbb{E}\left[KT(1+|x|^{2p})ds\right]<\infty,$$ (granted by assumption \eqref{ass:1}) it follows by dominated convergence that
					\begin{align*}
						\mathbb{E}\left[\int_t^{t'}|f(s,\Upsilon_s^t(x))|^2ds\right]\to 0,\quad \text{as}\quad t\to t'. 
					\end{align*}
					The generator $f$ regarded as real function in $y,\lambda,\sigma$ is continuous in all those variables (which follows from the form of $f$, the continuity of $\Phi$ and $\psi$, see Proposition \ref{prop: psicont}). As also the $z^t(x)$ are continuous in $x$ and $t$ by assumption \eqref{ass:2}, the generator is continuous in $t$ and $x$. Bounding 
					\begin{align*}
						\int_{t'}^T|f(s,\Upsilon_s^t(x))|^2ds+\int_{t'}^T\left|f(s,\Upsilon^t_s(x))-f(s,\Upsilon^{t'}_s(x'))\right|^2ds
					\end{align*}
					by $\int_0^T2KC\left(1+|x|^{2p}\right)ds$ makes dominated convergence applicable again, and we infer that $v$ is continuous in $t$ and $x$.

					As a last step, to give a sufficient relation to use the proof of \cite[Theorem 3.4]{bbp}, we have to bound the solution to the following BSDE:
					
					\begin{align*}
						\Upsilon_s^h=&\int_s^{t+h}\left(\tilde{\psi}(r,z_r^t(x))+f(r,\phi(r,z_r^t(x))+\Upsilon_r^h)\right)dr-\int_s^t\sigma_{\Upsilon,r}^h d\bar{W}_r\\
						&-\int_{(s,T]\times [0,l^{L}_{\max}]} U_{\Upsilon,r}^h(l)\tilde{\nu}(dr,dl),\quad s\in [t,t+h],
					\end{align*}
					where $\tilde{\psi}$ and $\phi$ are $C^2$-functions with polynomial growth. The same procedure as in Theorem \ref{Thm:ExistenceOfIndifferenceBSDESolutionUnbounded2Jumps} and the one in Theorem \ref{Thm:UniquenessOfIndifferenceBSDESolutionUnboundedJumps} shows existence and uniqueness of a solution $(\Upsilon^h,\sigma_{\Upsilon}^h,U_{\Upsilon}^h)\in \mathcal{S}^2\times L^2(\bar{W})\times L^2(\tilde{\nu})$ as the additive terms are well behaved. Further standard estimates, derived by It\^o's formula, yield
					\begin{align*}
						&\mathbb{E}\left[|\Upsilon_s^h|^2\right]+\frac{1}{2}\mathbb{E}\left[\int_s^{t+h}|\sigma_{\Upsilon,r}^h|^2dr+\int_s^{t+h}\int_{[0,l^{L}_{\max}]}|U_{\Upsilon,r}^h(l)|\vartheta(dl){dr}\right]\\
						&\leq \mathbb{E}\left[\int_s^{t+h}|\Upsilon_r^h|\left|\tilde{\psi}(r,z_r^t(x))+f\left(r,\phi(r,z_r^t(x))+\Upsilon_r^h\right)\right|dr\right].
					\end{align*}
					The polynomial growth of $\tilde{\psi}$, the bounds for the moments of $z_r^t(x)$ and the since $|f(s,y)|$ is bounded by $C(1+|\lambda(z^t_s(x))|+\sigma(z^t_s(x))^2)$ implies that there is a constant $c$ such that
					\begin{align*}
						&\mathbb{E}\left[|\Upsilon_s^h|^2\right]+\frac{1}{2}\mathbb{E}\left[\int_s^{t+h}|\sigma_{\Upsilon,r}^h|^2dr+\int_s^{t+h}\int_{[0,l^{L}_{\max}]}|U_{\Upsilon,r}^h(l)|\vartheta(dl){dr}\right]\\
						&\leq c\mathbb{E}\left[\int_s^{t+h}|\Upsilon_r^h|\left(1+|\Upsilon_r^h|+|\lambda(z^t_s(x))|+\sigma(z^t_s(x))^2\right)dr\right].
					\end{align*}
					Since $\lambda, \sigma$ grow at most polynomially and $z^t(x)$ obeys \eqref{ass:1}, it follows that
					$$\mathbb{E}\left[|\Upsilon_s^h|^2\right]\leq C\mathbb{E}\left[\int_s^{t+h}\left(|\Upsilon_r^h|+|\Upsilon_r^h|^2\right)dr\right],$$
					and, recalling that $\mathbb{E}\left[\sup_{s\in [0,T]}|Y^h_s|\right]<\infty$, we end up with
					$$\mathbb{E}\left[|\Upsilon_s^h|^2\right]\leq \tilde{C}h.$$
					From this point on, the proof can now be performed just as the one of \cite[Theorem 3.4]{bbp}.
				\end{proof}
				
								\begin{proof}(Proposition~\ref{Lem: CIR-cont})
					Inequality \eqref{eq:Cir1} follows by similar but easier calculations as in Lemma 2.1 of \cite{fujiwara1985} since the square root function fulfills the inequality $\sqrt{x}\le 1+x$. To show inequality \eqref{eq:Cir2}, we first define
					\begin{align*}
						z_r^t(x)-z_r^t(x')-(x-x') &= \int_{t}^{r}\!\kappa\left(z_u^t(x)-z_u^t(x')\right)du+\tilde{\varsigma}\!\int_{t}^{r}\left(\sqrt{z_u^t(x)}-\sqrt{z_u^t(x')}\right)d\hat{W}_u\\
						&=: A_r + M_r.
					\end{align*}
					Then,
					\begin{align*}
						\mathbb{E}\left[\sup_{ r \in [t,s]}\left|z_r^t(x)-z_r^t(x')-(x-x')\right|^p\right]\le C_p \left(\mathbb{E}\left[\sup_{ r \in [t,s]}\left|A_r\right|^p\right]+\mathbb{E}\left[\sup_{ r \in [t,s]}\left|M_r\right|^p\right]\right).
					\end{align*}
					Since $A$ has Lipschitz coefficients, we obtain
					\begin{align*}
						\mathbb{E}\left[\sup_{ r \in [t,s]}\left|A_r\right|^p\right]\le \kappa\left|t-s\right|^{p-1}\int_{t}^{s}\mathbb{E}\left[\left|z_u^t(x)-z_u^t(x')\right|^p\right]du
					\end{align*}
					by Jensen's inequality. With Doob's maximal inequality and It\^o's formula applied to $x\mapsto|x|^p$, we get constants $C_{p,1}, C_{p,2}>0$ (we will continue to number appearing $p$-dependent constants by $C_{p,i}$)
					\begin{equation}\label{eq:1}
						\begin{aligned}
							\mathbb{E}\left[\sup_{ r \in [t,s]}\left|M_r\right|^p\right]&\le C_p \mathbb{E}\left[\left|M_s\right|^p\right]\\
							&\le \tilde{C}_p\mathbb{E}\left[\int_{t}^{s}|M_u|^{p-2}\tilde{\varsigma}^2\left|\sqrt{z_u^t(x)}-\sqrt{z_u^t(x')}\right|^2du\right]
						\end{aligned}
					\end{equation}
					Next, we look at the Lamperti transformation. Therefore, we apply Lemma 3.2 from \cite{HJN} to $\sqrt{z_u^t(x)}$ (or $\sqrt{z_u^t(x')}$ respectively). We get
					\begin{equation}\label{squareroot}
						\begin{aligned}
							\sqrt{z_u^t(x)} &= \sqrt{x}+\int_{t}^{u}\left(\frac{4\kappa\theta-\tilde{\varsigma}^2}{8}\frac{1}{\sqrt{z_r^t(x)}}-\frac{\kappa}{2}\sqrt{z_r^t(x)}\right){1}_{\left\{\sqrt{z_r^t(x)}\in(0,\infty)\right\}}dr\\
							&\qquad+\frac{\tilde{\varsigma}}{2}\left(\hat{W}_u-\hat{W}_t\right)\\
							&=\sqrt{x}+\int_{t}^{u}g\left(\sqrt{z_r^t(x)}\right)dr+\frac{\tilde{\varsigma}}{2}\left(\hat{W}_u-\hat{W}_t\right)
						\end{aligned}
					\end{equation}
					where we set
					\begin{align*}
						g(x):=\left(\frac{\alpha}{x}+\beta x\right){1}_{\{x>0\}}
					\end{align*}
					with
					\begin{align*}
						\alpha = \frac{4\kappa\theta-\tilde{\varsigma}^2}{8}, \qquad \beta=-\frac{\kappa}{2}.
					\end{align*}
					Note that $g$ satisfies the following inequality for $\alpha>0 ,x,x'\ge0$ and $\beta\in\mathbb{R}$:
					\begin{equation}\label{f_ineq}
						(x-x')(g(x)-g(x'))\le \beta(x-x')^2 + \alpha\left({1}_{\{x=0\}}+{1}_{\{x'=0\}}\right)
					\end{equation}
					without the second term on the right side, this would be the so-called one-sided Lipschitz continuity (see e.g.~\cite{DNS}).
					From \eqref{squareroot}, we have
					\begin{align*}
						\sqrt{z_u^t(x)}-\sqrt{z_u^t(x')}=\sqrt{x}-\sqrt{x'}+\int_{t}^{u}g\left(\sqrt{z_r^t(x)}\right)-g\left(\sqrt{z_r^t(x')}\right)dr
					\end{align*}
					and by Ito's formula, we get
					\begin{align*}
						\left(\sqrt{z_u^t(x)}-\sqrt{z_u^t(x')}\right)^2&=\left(\sqrt{x}-\sqrt{x'}\right)^2\\
						&+2\int_{s}^{u}	\left(\sqrt{z_r^t(x)}-\sqrt{z_r^t(x')}\right)\left(g\left(\sqrt{z_r^t(x)}\right)-g\left(\sqrt{z_r^t(x')}\right)\right)dr.
					\end{align*}
					Since $g$ fulfills inequality \eqref{f_ineq}, we obtain
					\begin{align*}
						\left(\sqrt{z_u^t(x)}-\sqrt{z_u^t(x')}\right)^2\le \left(\sqrt{x}-\sqrt{x'}\right)^2 + 2\alpha \int_{t}^{u} \left({1}_{\left\{\sqrt{z_r^t(x)}=0\right\}}+{1}_{\left\{\sqrt{z_r^t(x')}=0\right\}}\right) dr
					\end{align*}
					since $\beta <0$. Furthermore, note that the second term on the right is $0$ $\mathbb{P}$-a.s.~by Lemma 3.1 in \cite{HJN}. Inserting this into \eqref{eq:1} and applying Young's inequality, we get 
					\begin{align*}
						\mathbb{E}\left[\sup_{ r \in [t,s]}\left|M_r\right|^p\right]
						&\le C_{p,3}\mathbb{E}\left[\int_{t}^{s}|M_u|^{p-2}\left(\sqrt{x}-\sqrt{x'}\right)^2du\right]\\
						&\le C_{p,4}\mathbb{E}\left[\int_{t}^{s}|M_u|^{p}+\left|\sqrt{x}-\sqrt{x'}\right|^pdu\right]\\
						&\le C_{p,5}\left((s-t)\left|\sqrt{x}-\sqrt{x'}\right|^p+\int_{t}^{s}\mathbb{E}\left[\sup_{ r \in [t,u]}\left|M_r\right|^p\right]du\right).
					\end{align*}
					Applying Gronwall's inequality gives
					\begin{align*}
						\mathbb{E}\left[\sup_{ r \in [t,s]}\left|M_r\right|^p\right]\le C_{p,6}(s-t)\left|\sqrt{x}-\sqrt{x'}\right|^p
					\end{align*}
					Summarizing, we have
					\begin{align*}
						&\mathbb{E}\left[\sup_{ r \in [t,s]}\left|z_r^t(x)-z_r^t(x')-(x-x')\right|^p\right]\\
						&\le C_{p,7}\left(\int_{t}^{s}\mathbb{E}\left[\left|z_u^t(x)-z_u^t(x')\right|^p\right]du+(s-t)\left|\sqrt{x}-\sqrt{x'}\right|^p\right)\\
						&\le C_{p,8}\left((s-t)\left|x-x'\right|^p+(s-t)\left|\sqrt{x}-\sqrt{x'}\right|^p\right.\\
						&\left.\qquad\qquad+\int_{t}^{s}\mathbb{E}\left[\sup_{ r \in [t,u]}\left|z_r^t(x)-z_r^t(x')-(x-x')\right|^p\right]du\right).
					\end{align*}
					Applying Gronwall's inequality finishes the proof of \eqref{eq:Cir2}. To prove \eqref{L1_cir}, we can assume $x\ge x'$ without loss of generality. Then,
					\begin{align*}
						\mathbb{P}\left(z_t(x)\ge z_t(x'), \text{for all }0\le t <\infty\right)=1
					\end{align*}
					by Proposition 5.2.18 in \cite{KS}. Then, we have
					\begin{align*}
						\mathbb{E}\left[\left|z_t(x)-z_t(x')\right|\right]=\mathbb{E}\left[z_t(x)-z_t(x')\right]=(x-x')e^{-\kappa t}
					\end{align*}
				since $z_t(x)$ is non-central chi-square distributed with expectation $\theta+(x-\theta)e^{-\kappa t}$.
				\end{proof}

\begin{proof}(Proposition \ref{prop:CIR-exp-integr})\\
By \eqref{CIR}, we get that for a starting value $x$ of $z$,
\begin{align*}
\exp(\varepsilon z_t)=\exp\bigg(\varepsilon x+\varepsilon\kappa\theta t-\varepsilon\kappa\int_0^tz_sds+\varepsilon\tilde{\varsigma}\int_0^t\sqrt{z_s}d\hat{W}_s\bigg).
\end{align*}
Proposition 3.2 in \cite{CozReis2016} states that
\begin{align*}
\sup_{t\in [0,T]}\mathbb{E}\bigg[-\varepsilon\kappa\int_0^tz_sds+\varepsilon\tilde{\varsigma}\int_0^t\sqrt{z_s}d\hat{W}_s\bigg]<\infty,\quad\text{if}\quad-\kappa\varepsilon+\frac{\varepsilon^2\tilde{\varsigma}^2}{2}<0,
\end{align*}
which is the case if $\varepsilon$ is chosen smaller than $\frac{2\kappa}{\tilde{\varsigma}^2}$.
\end{proof}
		\end{document}